\documentclass[conference]{IEEEtran}

\IEEEoverridecommandlockouts

\usepackage{amsfonts,amsmath,amssymb,bbm, amsthm}
\usepackage{graphicx,epsfig,color}
\usepackage{algorithm,algorithmic}
\usepackage{comment}
\usepackage{epstopdf}
\usepackage{bm}
\usepackage[font=footnotesize]{caption}
\usepackage{subfigure}
\usepackage{nicefrac}
\usepackage{multirow}
\usepackage{array}
\usepackage{booktabs}
\usepackage{enumerate}

\newtheorem{thm}{Theorem}

\newtheorem{lem}{Lemma}

\newtheorem{coro}{Corollary}

\newtheorem{assume}{Assumption}

\newcommand*{\Scale}[2][4]{\scalebox{#1}{$#2$}}%

\def\BibTeX{{\rm B\kern-.05em{\sc i\kern-.025em b}\kern-.08em
		T\kern-.1667em\lower.7ex\hbox{E}\kern-.125emX}}
\begin{document}

\title{Optimal Dynamic Orchestration in NDN-based Computing Networks}

\author{
\IEEEauthorblockN{
Hao Feng, Yi Zhang, Srikathyayani Srikanteswara, Marcin Spoczynski, Gabriel Arrobo,\\ Jing Zhu, Nageen Himayat
}
\IEEEauthorblockA{
Intel Labs, USA, Email: \{hao.feng, yi1.zhang, srikathyayani.srikanteswara, marcin.spoczynski, gabriel.arrobo,\\ jing.z.zhu, nageen.himayat\}@intel.com
}

}


\maketitle

\begin{abstract}
Named Data Networking (NDN) offers promising advantages in deploying next-generation service applications over distributed computing networks. We consider the problem of dynamic orchestration over a NDN-based computing network, in which nodes can be equipped with communication, computation, and data producing resources. Given a set of services with function-chaining structures, we address the design of distributed online algorithm that controls each node to make adaptive decisions on flowing service requests, committing function implementations, and/or producing data. We design a Service Discovery Assisted Dynamic Orchestration (SDADO) algorithm that reduces the end-to-end (E2E) delay of delivering the services, while providing optimal throughput performance. The proposed algorithm hybrids queuing-based  flexibility and topology-based discipline, where the topological information is not pre-available but obtained through our proposed service discovery mechanism. We provide throughput-optimality analysis for SDADO, and then provide numerical results that confirm our analysis and demonstrates reduced round-trip E2E delay. 
\end{abstract}

\pagenumbering{arabic}

\IEEEpeerreviewmaketitle


\section{Introduction}
\label{sec_introduction}
Next-generation network service applications with high throughput and low latency requirements, 
such as online gaming, real-time remote sensing, augmented reality,
etc., are explosively dominating the internet traffic. 
On the other hand, with the significant growth in both aspects of quantity and performance, edge devices provide promising capability of tackling computing tasks, especially if we extend the edge resources to include far edge devices, e.g., cameras, IoT devices, mobile devices, and vehicles. 
Therefore, elastically dispersing computing workload to edge gradually lead the solutions to deploying the next-generation network services. However, with arbitrary network scale and topology, three questions are naturally raised: i) where to execute the network functions and produce source data; ii) how to steer the data flowing toward appropriate compute and data resources; iii) how to dynamically adapt the decisions to changing service demands. To solve these issues, centralized solutions have been extensively explored \cite{bib_heft}\cite{bib_Yang_et_al}, but have been challenged due to their scalability and cost \cite{bib_Sarkar_et_al}\cite{bib_Bonomi_et_al}, which in turn makes distributed solutions appealing. 

In distributed and dynamic in-network-computing paradigm, it has been found that Named Data Networking (NDN)-based computing orchestration framework offers advantages over traditional edge computing orchestration \cite{bib_globecom_Kathyayani}, where NDN obviates having a central registry that tracks all resources. In context of NDN, resources can be reached using \emph{name}-based queries regardless their locations. When network topology is changing, each device can send a new query with the name remaining constant even if the device's address changes. With these benefits, NDN overlays over IP are being explored as a promising solution for computing networks \cite{bib_ndn_over_ip}. 

Existing popular NDN routing schemes, e.g. Best Route, Multi-cast, focus on helping nodes to make forwarding decisions, and the target network is only for data delivery \cite{bib_nfd_guider}. Although modifications can be done to extend these algorithms to computing networks, they still lack enough adaptability to time-varying traffics. On the other hand, dynamic network control policies have been studied in \cite{bib_DCNC}, where the proposed backpressure-based algorithm DCNC has shown promising adaptability to dynamic environment but lacks routing discipline that results in compromised delay performance. An enhanced algorithm EDCNC, which biases the backpressure routing using topological information, is further proposed in Ref. \cite{bib_DCNC} and shows promising end-to-end (E2E) latency. However, EDCNC requires the whole network to apply a global bias parameter, whose value significantly influences the delay performance but has no clear clue to choose.   
The gap between NDN and backpressure-based dynamic orchestration has been firstly addressed in \cite{bib_DECO} on orchestrating single-function network services, where the proposed scheme DECO allows each node to make adaptive decisions on whether or not to accept NDN computing requests. Nevertheless, routing is assumed to be predetermined in \cite{bib_DECO}, which may not be the case in many networking scenarios.

In this paper, we propose a Service Discovery Assisted Dynamic Orchestration (SDADO) algorithm for deploying \emph{function-chaining} services in NDN-based distributed computing networks with the following contributions: i) we develop a SDADO orchestration framework for committing NDN-based chaining requests for compute/data; ii) we propose a name-based \emph{service discovery} mechanism to search for topological information processed with targeted function-chaining structures; iii) we propose a novel SDADO orchestration algorithm that achieves reduced E2E latency and throughput-optimality, and provide theoretical analysis and numerical evaluations. 

The paper is as organized follows. Section \ref{sec_system_model} presents the system model. Section \ref{sec_service_discovery} describes the service discovery mechanism. The SDADO algorithm's description, analysis, and numerical evaluation are respectively presented in Section \ref{sec_sdado}, \ref{sec_analysis}, and \ref{sec_simulation}. The paper is concluded in Section \ref{sec_conclusion}.

\section{System Model}
\label{sec_system_model}
\subsection{Network and Service Model}
\label{subsec_network_service_model}
We consider a time-slotted network system with slots normalized to integer units $t\!\in\!\{0,1,2\cdots\}$. A computing network $\mathcal G=(\mathcal V,\mathcal E)$ with arbitrary topology consists of $\left|\mathcal V\right|$ nodes inter-connected by $\left|\mathcal E\right|$ communication links. Each node $i$ in the network represents a network unit, e.g., user equipment, access point, edge server, or data center, etc., and let $(i,j)$ represent the link from node $i$ to $j$. Denote $\mathcal O(i)$ as the set of neighbor nodes having links incident to node $i$. Let $C_{ij}(t)$ represent the maximum transmission rate over link $(i,j)$ at time $t$, whose value evolves ergodically, and its statistical average $\bar C_{ij}$ can be estimated by averaging $C_{ij}(t)$ over time.

A network service $\phi$ is described by a chain of $K_\phi$ functions plus the data producing. We use the pair $(\phi,k)$ with $k\in \{0,\cdots ,K_\phi \} $ to denote the $k$-th function of service $\phi$, where we index the data producing of service $\phi$ by $(\phi, 0)$ and call it \emph{function $(\phi, 0)$}. We denote by $r^{(\phi,k)}$ the \emph{processing complexity} of implementing function $(\phi,k)$, where $1\le k \le K_\phi$. That is, when a data packet is processed by implementing function $(\phi,k)$, it consumes $r^{(\phi,k)}$ units of computing cycles. Denote $\mathcal C_\phi$ as the set of service $\phi$'s consumers; denote ${\mathcal P}_{(\phi,k)} $ as the set of nodes that are equipped with compute or data resources to execute function $(\phi,k)$, $0\le k \le K_\phi$; denote by $C_{i}^{\text{pr}}$ the processing capacity (with unit of, e.g., cycles/slot) at node $i\in \mathcal P_{(\phi,k)}$, $1\le k \le K_\phi$, and by $C_i^{\text{dp}}$ as the data producing capacity at node $i\in \mathcal P_{(\phi,0)}$. 

In context of NDN, the deployment of a streaming service $\phi$ is a process of delivering interest packet streams carrying requests from a consumer to one or multiple data producers, followed by a process of delivering back the data packet streams. Each data packet is flowing through the reverse path traveled by its corresponding interest packet. In order to guarantee that data is processed in the function-chaining order, interest packets carrying the requests flow through resource-equipped nodes that \emph{commit} implementing the service's functions in the reversed chaining order, i.e., $(\phi,K_\phi)\rightarrow \cdots \rightarrow (\phi,0)$. 

Streams of interest/data packets delivered for a service can be modeled into chaining stages as is shown in Fig. \ref{fig_function_chain}. We define each stage of interest packets serving consumer $c\in \mathcal C_{\phi}$ requesting for a function $(\phi, k)$ by an \emph{interest commodity} and use $(\phi,k,c)$ to index it, and with the same index, define each stage of data packets output from function $(\phi, k)$ serving consumer $c$ by a \emph{data commodity}. In NDN paradigm, the commodity index of an interest/data packet can be inferred from its name, and a \emph{processing commitment} can updates an interest packet's commodity by updating it's name. 
We let $z^{(\phi,k)}$, $0\le k \le K_\phi$, represent the size of a data commodity $(\phi,k,c)$ packet for all $c\in \mathcal C_\phi$, and assume that the network traffic is dominated by the data packets because the size of an interest packet is negligible respective to $z^{(\phi,k)}$.
\begin{figure}
	\centering
	\includegraphics[width=3.5in]{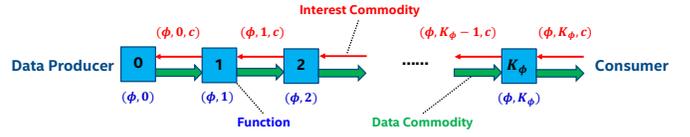}
	\caption{Multi-commodity chain flow model under NDN.}
	\label{fig_function_chain}
	\vspace{-0.1cm}
\end{figure}
\begin{figure}
	\centering
	\includegraphics[width=3in]{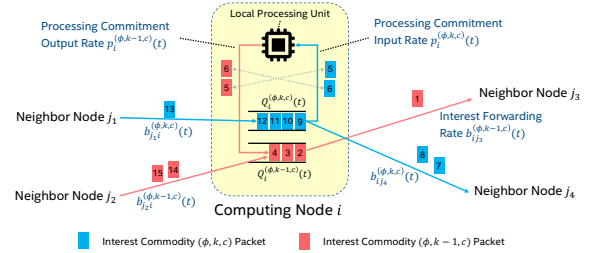}
	\caption{Interest packet queues' update at a computing node.}
	\label{fig_queuing_model}
	\vspace{-0.5cm}
\end{figure}

\subsection{Queuing Model}
\label{subsec_queuing_model}
We define by $a_c^{\phi}(t)$ the number of requests for service $\phi$ generated by consumer $c$ at time $t$ and by $\lambda_c^{\phi}$ its expected value. We assume that $a_c^{\phi}(t)$ is independently and identically distributed (i.i.d) across timeslots, and $\sum\nolimits_{\phi}a_i^\phi(t) \le A_{\max}$. At each time $t$, every node buffers the received interest and data packets into queues respectively according to their commodities. Each interest/data queue builds up from the transmission of interest/data packets from neighbors, request/data generations, and/or local processing-commitment/processing. Each node makes local orchestration decisions by controlling the queuing of interest packets, while the data packets are forwarded backward and get processed as committed according to the Pending Interest Table (PIT), and their queues are simply updated in First-In-First-Out mode. Fig. \ref{fig_queuing_model} shows an example that a computing node equipped with function $(\phi,k)$ makes decisions on whether or not to commit processing for the arrived interest commodity $(\phi,k,c)$ packets. Once a commitment is made, an interest packet is transported from commodity $(\phi,k,c)$ queue to $(\phi,k-1,c)$ queue with the its commodity index updated.

We define $Q_i^{(\phi,k,c)}(t)$ as the number of interest commodity $(\phi,k,c)$ packets queued in node $i$ at the beginning of timeslot $t$; define $b_{ij}^{(\phi,k,c)}(t)$ as the assigned number of interest commodity  $(\phi,k,c)$ packets to forward over link $(i,j)$ in timeslot $t$; define $p_{i}^{(\phi,k,c)}(t)$ as the assigned number of interest commodity $(\phi,k,c)$ packets to commit in timeslot $t$ for function $(\phi,k)$. Then node $i$ has the following queuing dynamics:
\begin{align}
&\Scale[0.95]{Q_i^{(\phi,k,c)}(t+1) \!\!\le \!\!
 \left[\!Q_i^{(\phi,k,c)}(t) -\!\! \sum_{j\in \mathcal O(i)}{b_{ij}^{(\phi,k,c)}}(t) - p_{i}^{(\phi,k,c)}(t) \!\right]^{\!+}}\notag\\
 &\Scale[0.95]{+\!\!\sum_{j\in \mathcal O(i)}{b_{ji}^{(\phi,k,c)}}(t) + p_{i}^{(\phi,k+1,c)}(t) + a_i^{\phi}(t)\mathbbm{1}\!\left(k=K_\phi\right),}
 \label{eq_queue_dynamic}
\end{align}   
where $[*]^+$ represents $\max\{*,0\}$; $a_i^{\phi}(t)\!\!=\!\!0$ if $i\!\notin\! \mathcal C_\phi$; $p_{i}^{(\phi,k,c)}(t)\!\!=\!\!0$ if $k\!\!>\!\!K_\phi$ or $i\!\!\notin\!\! \mathcal P_{(\phi,k)}$; 
$\mathbbm{1}(*)$ is the indicator function of $*$;
$p_{i}^{(\phi,k+1,c)}(t)$ is the number of interest commodity $(\phi,k+1,c)$ packets assigned to be transported into commodity $(\phi,k,c)$ queue after being committed by function $(\phi,k+1)$ in timeslot $t$.


\subsection{System Constraints}
\label{subsec_system_constraints}
The orchestration dynamically controls the action vector $[\{b_{ij}^{(\phi,k,c)}(t)\}, \{p_i^{(\phi,k,c)}(t)\}]$ subject to the constraints summarized in \eqref{eq_system_constraints}:
\begin{subequations}\label{eq_system_constraints}
	\begin{align}
		& \lim_{t\rightarrow \infty}\nicefrac{Q_i^{(\phi,k,c)}(t)}{t}=0\ \text{with prob. 1},\ \forall i,\phi,k,c,\label{eq_stability1}\\
		& \Scale[0.95]{\sum\nolimits_{(\phi,k,c)}{z^{(\phi,k)}b_{ij}^{(\phi,k,c)}(t)} \le \bar C_{ji}, \ \ \ \forall (i,j)\in \mathcal E,t,}\label{eq_comm_capacity}\\
		& \Scale[0.95]{\sum\nolimits_{(\phi,k>0,c)}{r^{(\phi,k)}p_i^{(\phi,k,c)}(t)} \le C_i^{\text{pr}}, \ \ \forall i\in \mathcal P_{(\phi,k)},t,}\label{eq_compute_capacity}\\
		& \Scale[0.95]{\sum\nolimits_{(\phi,0,c)}{z^{(\phi,0)}p_i^{(\phi,0,c)}(t)} \le C_i^{\text{dp}}, \ \ \forall i\in \mathcal P_{(\phi,0)}, t,}\label{eq_produce_capacity}\\
		& \Scale[0.95]{b_{ij}^{(\phi,k,c)}\!(t)\!\in\! \mathbb{Z}^+, p_i^{(\phi,k,c)}\!(t)\!\in\! \mathbb{Z}^+,\ \forall \phi,k,c,i,(i,j)\!\in \!\mathcal E, t.}\label{eq_value_domain}
	\end{align}
\end{subequations}
where \eqref{eq_stability1} characterizes the \emph{network rate stability} \cite{bib_Neely_book}; \eqref{eq_comm_capacity}-\eqref{eq_produce_capacity} present the capacity constraints for interest forwarding, processing commitment, and data producing.
Note that each interest packet flowing over a link eventually results in a data packet flowing backward over the reverse link, and considering that data flow dominates the traffic (see Section \ref{subsec_network_service_model}), we have \eqref{eq_comm_capacity} showing that interest packets' flowing over link $(i,j)$ reflects data packets' backward flowing constrained by the expected capacity of link $(j,i)$.

\begin{figure}
	\centering
	\hspace{-0.3cm}
	\subfigure[]{
		\includegraphics[width=4cm]{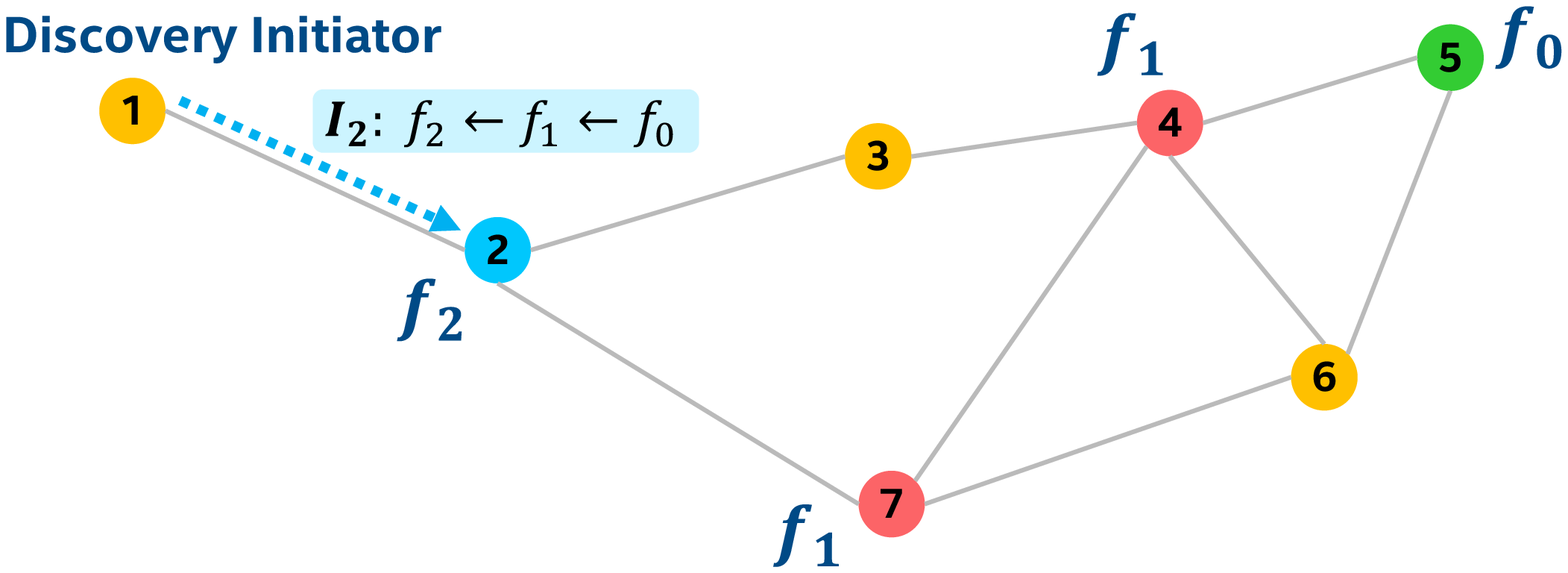}
		\label{fig_discovery1}
	}
	\hspace{-0.3cm}
	\subfigure[]{
		\includegraphics[width=4cm]{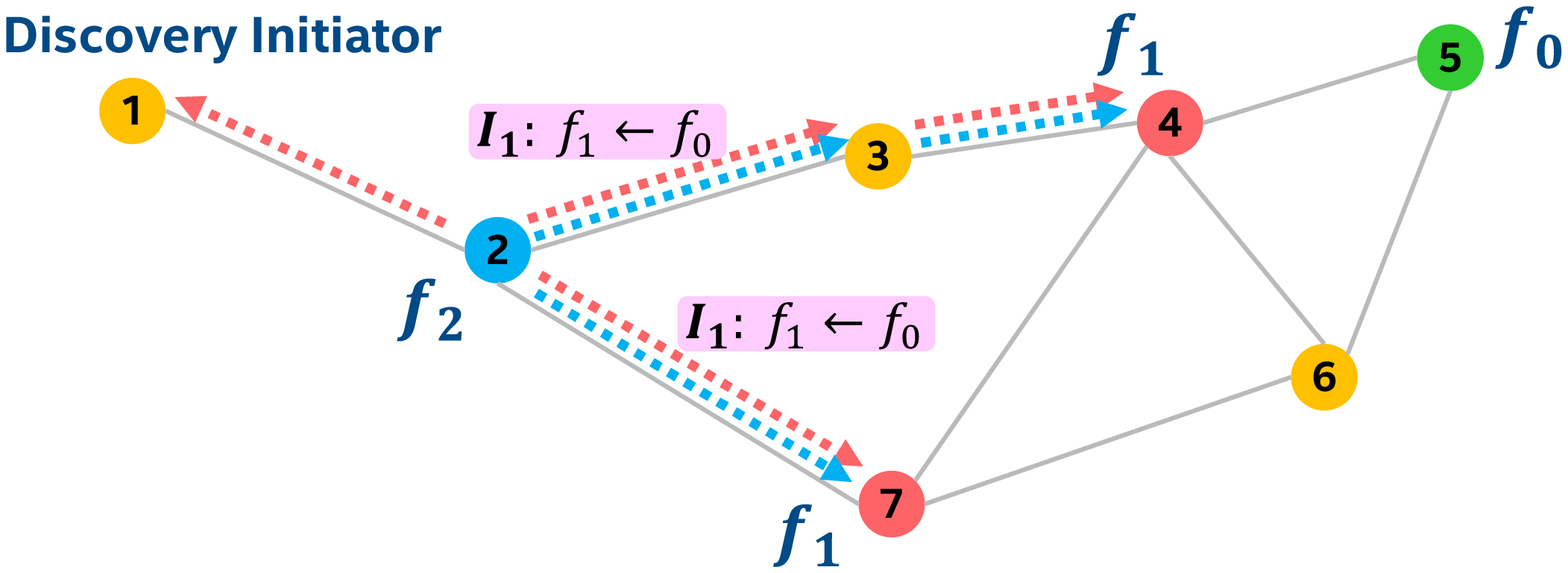}
		\label{fig_discovery2}
	}
	\hspace{-0.3cm}
	\subfigure[]{
		\includegraphics[width=4cm]{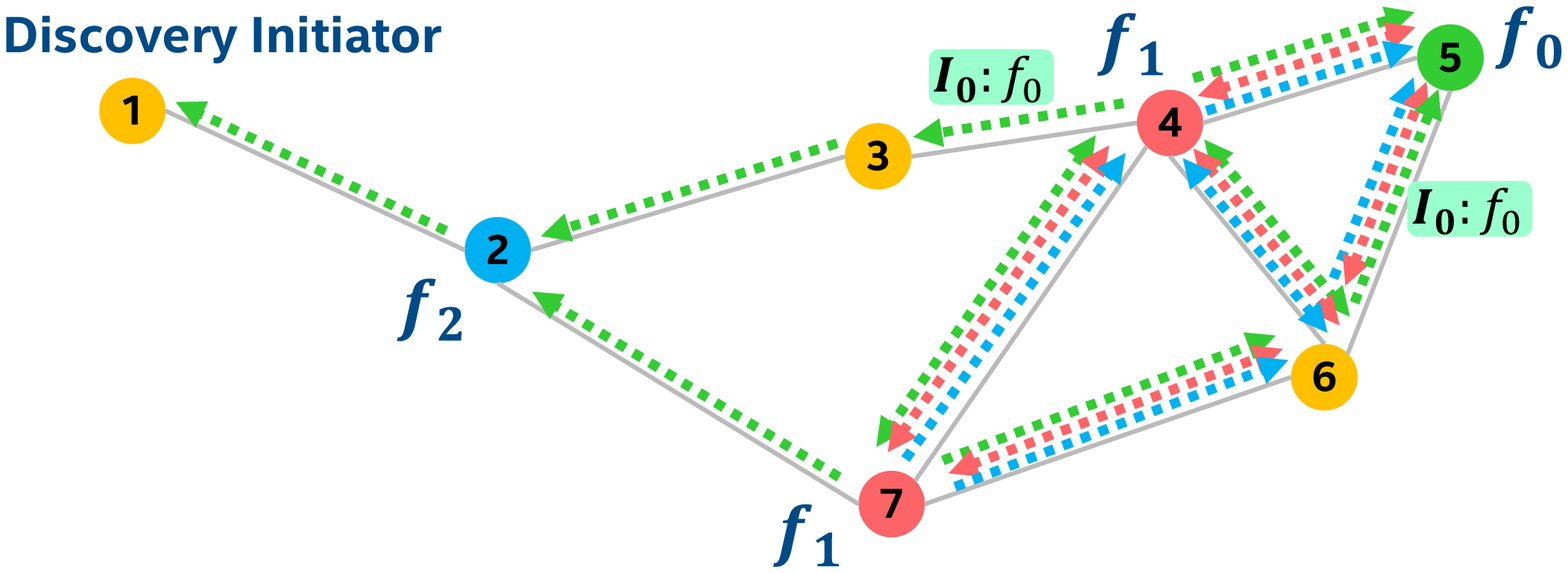}
		\label{fig_discovery3}
	}
	\hspace{-0.3cm}
	\subfigure[]{
		\includegraphics[width=4cm]{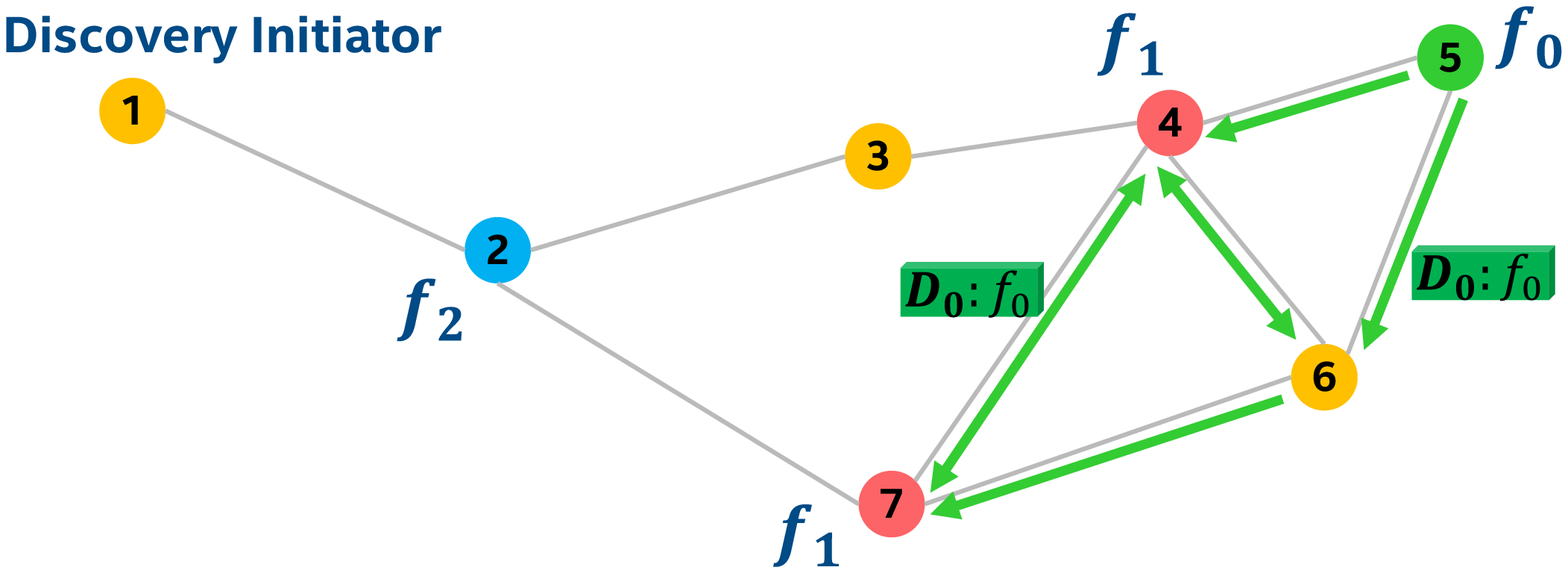}
		\label{fig_discovery4}
	}
	\hspace{-0.3cm}
	\subfigure[]{
		\includegraphics[width=4cm]{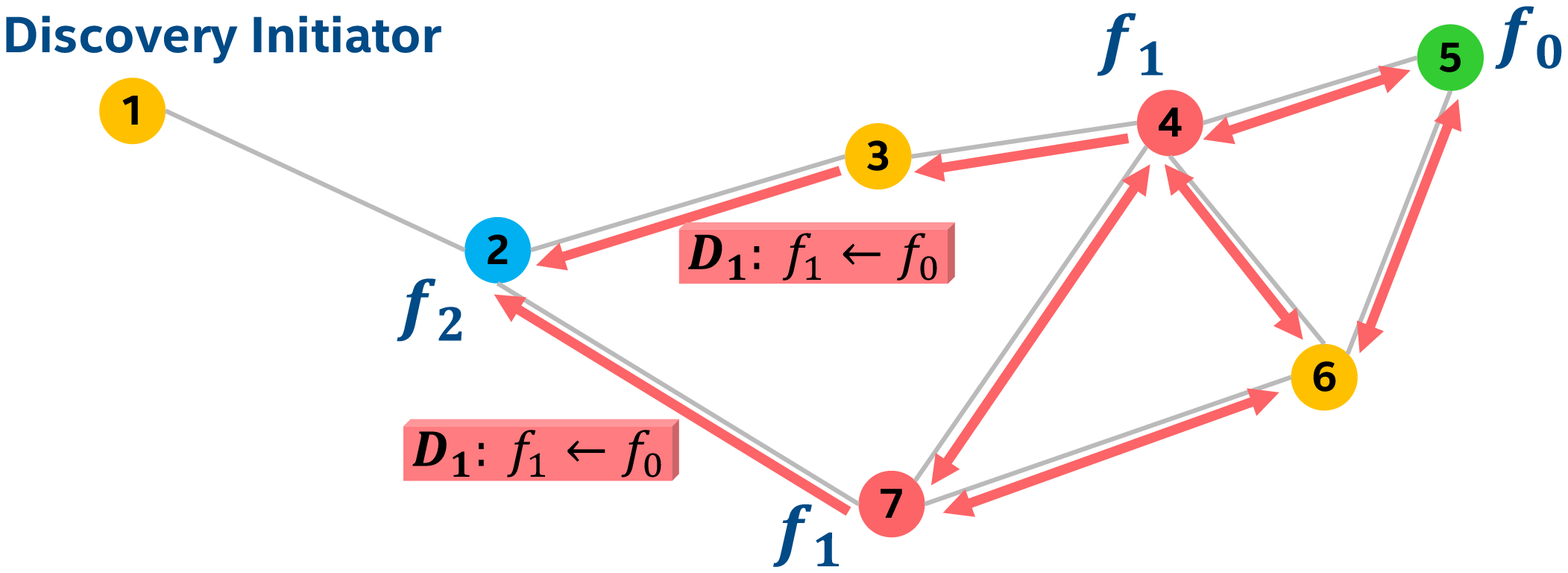}
		\label{fig_discovery5}
	}
	\hspace{-0.3cm}
	\subfigure[]{
		\includegraphics[width=4cm]{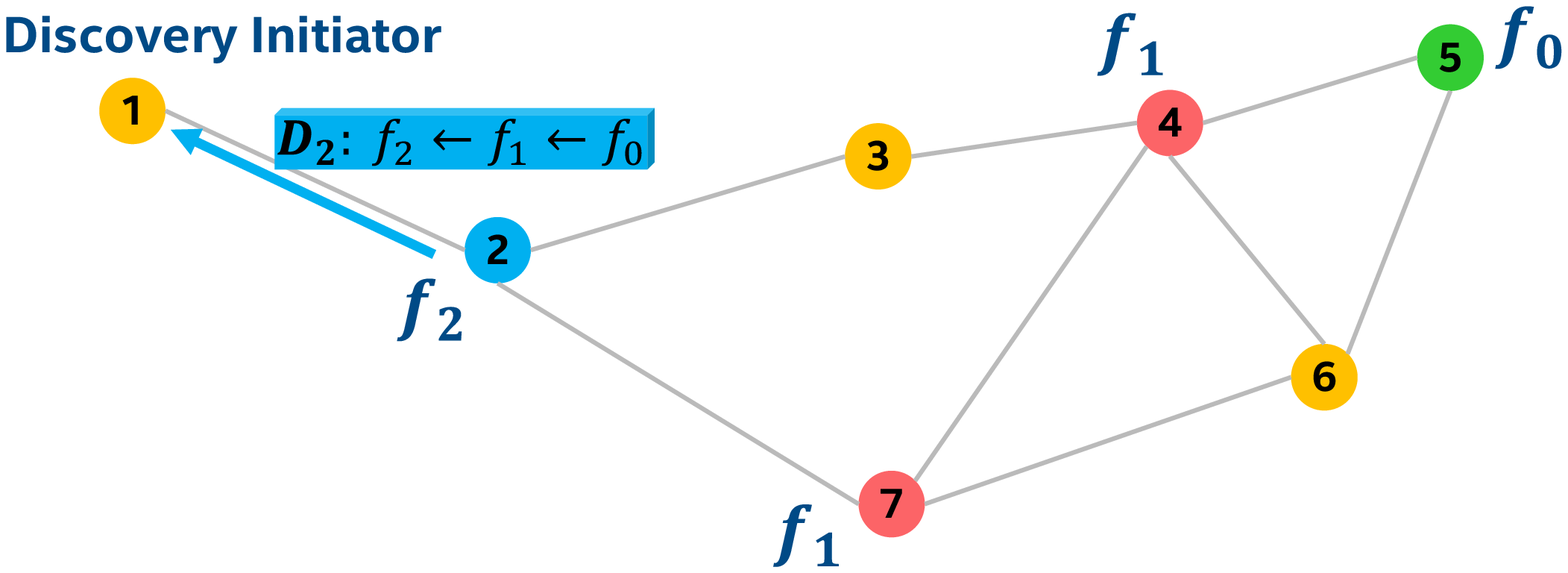}
		\label{fig_discovery6}
	}
	\caption{
		A service discovery example with targeted function chain segment $f_0, f_1, f_2$:
		a)~a discovery interest packet $I_2$ is generated and propagated until reaching node $2$ having function $f_2$; b)~node $2$ generates and send out discovery interest packets $I_1$, $I_2$; c)~generate discovery interest packets $I_0$ and forward $I_0$, $I_1$, $I_2$; d)~generate and forward discovery data packets $D_0$; e)~generate and forward discovery data packets $D_1$; f)~generate and forward discovery data packets $D_2$.
	}
	\vspace{-0.5cm}
	\label{fig_discovery_example}
\end{figure}

\section{Service Discovery}
\label{sec_service_discovery}
Although knowledge of compute/data resource distribution and network topology is helpful for making efficient orchestration decisions, efficiently fetching this knowledge in an arbitrary distributed network can be challenging. In this paper, we propose that, to prepare deploying a targeted function-chain segment, a node can initiate a \emph{function-chaining service discovery} to obtain an estimate of the minimum \emph{abstract distances} that the interest/data packets will travel via each face to reach through the function chain segment. 
In a multi-hop computing network, a function-chaining service discovery is composed by i) a phase of generating and propagating \emph{discovery interest packets} that sequentially search the targeted functions in the reverse chaining order and ii) a phase of sending back \emph{discovery data packets} gathering, carrying and updating the minimum abstract distance values. 

\subsection{Searching via Discovery Interest Packets}
\label{subsec_discovery_interest}
For the targeted function chain segment $(\phi,0), \cdots, (\phi,k)$, $0\le k \le K_\phi$, the discovery interest packets search the resources in order $(\phi,k)\rightarrow \cdots \rightarrow (\phi,0)$, where function $(\phi,k)$ is the \emph{immediate searching target}. Each discovery interest packet is forwarded by multi-cast. In context NDN, the name of a discovery interest packet $p$ can be $``(\phi,k) \leftarrow \cdots \leftarrow (\phi,0)"$. Once a node having function $(\phi,k)$ is reached by packet $p$, a new discovery interest packet $p'$ with name $``(\phi,k-1) \leftarrow \cdots \leftarrow (\phi,0)"$ is generated and multi-cast out to search for function $(\phi,k-1)$. In the meanwhile, packet $p$ is still forwarded out searching for other nodes having function $(\phi,k)$ if its hop limit is not reached, with the purpose of letting the current node know where to offload function $(\phi,k)$'s tasks in the service's deployment if needed. Fig. \ref{fig_discovery1}-\ref{fig_discovery3} present the searching phase of a service discovery example.

Fig. \ref{fig_discovery_example} shows service discovery for an example function segment $f_0, f_1, f_2$. Node 1 in Fig. \ref{fig_discovery1} initiates the discovery by generating a discovery interest packet $I_2$ with name $``f_2\leftarrow f_1\leftarrow f_0"$, and $I_2$ flows to node $2$ having function $f_2$. In Fig. \ref{fig_discovery2}, node 2 generates a new discovery interest packet $I_1$ with name $``f_1\leftarrow f_0"$, and copies of $I_1$ propagate and reach nodes 4 and 7 having function $f_1$; $I_2$ copies are also forwarded. Fig. \ref{fig_discovery3} shows that both node 4 and 7 generate the packets $I_0$ with name $``f_0"$, and copies of $I_0$ propagate and reach the data producer (node 5). In addition of $I_2$'s forwarding, node 4 and 7 forward $I_1$ copies that reach each other to prepare for their mutual discovery. 

With the purpose of detecting multi-paths, the \emph{loop detection} used in service discovery is different from the existing mechanism that is used in the standard NDN (see Ref. \cite{bib_nfd_guider}) based on detecting nonce number collision. The nonce number of each discovery interest packet is updated at each hop to avoid nonce duplication. A loop detection in service discovery enables the discovery interest packet to carry a \emph{stop list} that records the sequence of IDs of the traveled nodes. If a node having received a discovery interest packet detects that its own node-ID belongs to the packet's stop list, a looping is identified.


\subsection{Feedback via Discovery Data Packets}
\label{subsec_discovery_data}
A discovery data packet with name $``(\phi,k)\leftarrow\cdots\leftarrow (\phi,0)"$ is generated in either of the two cases: i) for $k\ge 1$, a node having function $(\phi,k)$ has received/gathered one or multiple discovery data packets with name $``(\phi,k-1)\leftarrow\cdots\leftarrow (\phi,0)"$; ii) for $k = 0$, a node having function $(\phi,0)$, i.e., data producer, has received a discovery interest packet with name $``(\phi,0)"$. Each discovery data packet flows backward, i.e., if a discovery data packet with name $``(\phi,k)\leftarrow\cdots\leftarrow (\phi,0)"$ flows over link $(j,i)$, a discovery interest packet with the same name should have flowed over link $(i,j)$. Fig. \ref{fig_discovery4}-\ref{fig_discovery6} present the feedback phase of the service discovery example.

For the example in Fig. \ref{fig_discovery_example}, the data producer (node 5) in Fig. \ref{fig_discovery4} generates copies of discovery data packet $D_0$ with name $``f_0"$ that flow backward until reaching node 4 and 7. In Fig. \ref{fig_discovery5}, node 4 and 7 generate the discovery data packets $D_1$ with name $``f_1\leftarrow f_0"$, and $D_1$ copies not only flow to node 2 but also flow between node 4 and 7 as their mutual discovery feedback. Fig. \ref{fig_discovery6} shows that node 2 generates discovery data packet $D_2$ with name $``f_2\leftarrow f_1\leftarrow f_0"$ that flows back to the discovery initiator (node 1). 
\begin{algorithm}
	\caption{Update of the minimum abstract distances $L_i^{(\phi,k)}$
	}
	\label{alg_discovery}
	\begin{algorithmic}[1]
		\STATE Initialize $L_i^{(\phi,k)}=L_j^{(\phi,k)}=+\infty$ for all $0\le k \le K_\phi$.
		\FOR {$k=0$ to $K_\phi$}
		\IF {$k=0$ and node $i\in \mathcal P_{(\phi,0)}$}
		\STATE $L_i^{(\phi,0)}\gets l_{i}^{\text{dp}}$.
		\ELSIF{$k>0$ and node $i\in \mathcal P_{(\phi,k)}$}
		\STATE $\Scale[0.95]{L_i^{(\phi,k)}\!\gets\!\min\{l_{i}^{\text{pr}} \!+\! L_i^{(\phi,k-1)}, \min\nolimits_{j\in \mathcal O(i)}\{l_{ij}\!+\!L_j^{(\phi,k)}\}\!\}}$. 
		\ELSE
		\STATE $L_i^{(\phi,k)}\gets \min\nolimits_{j\in \mathcal O(i)}\!\{l_{ij}+\!L_j^{(\phi,k)}\}$.
		\ENDIF
		\ENDFOR
	\end{algorithmic}
\end{algorithm}
\begin{figure}
	\centering
	\includegraphics[width=3in]{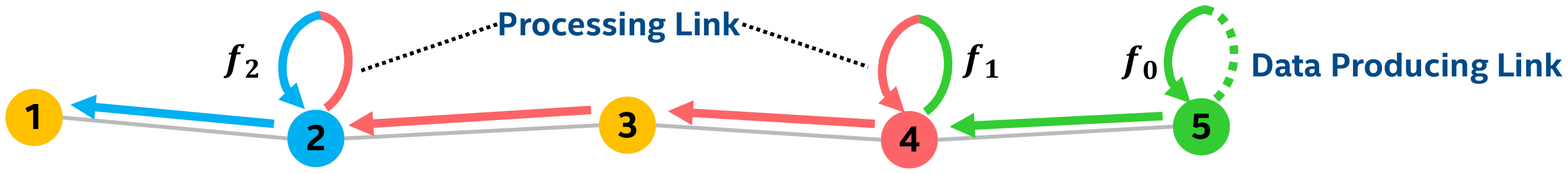}
	\caption{An example of augmented path.}
	\label{fig_augmented_path}
	\vspace{-0.5cm}
\end{figure}

We further model implementing a function (or data producing) at a node by going through a “processing link” (or “data producing link”).
As is shown in Fig. \ref{fig_augmented_path}, adding the processing links and data producing links onto a feedback routing path forms an \emph{augmented path}, which characterizes the trajectory traveled by a relaying series of data packets going through a function chain segment. These discovery data packets gather, update and deliver the minimum abstract distance values along the augmented path, and the path's abstract distance is calculated by summing the \emph{abstract lengths} of all the links on it. We measure an abstract link length as follows: i) $l_{ij}\triangleq \nicefrac{1}{\bar C_{ji}}$ for a communication link; ii) $l_{i}^{\text{pr}} = \nicefrac{r^{(\phi,k)}}{z^{(\phi,k)}C_i^{\text{pr}}}$ for a processing link; iii) $l_{i}^{\text{dp}}\triangleq \nicefrac{1}{C_i^{\text{dp}}}$ for a data producing link.

We denote by $L_i^{(\phi,k)}$ the minimum abstract distance to travel to node $i$ going through function chain segment $(\phi,0),\cdots,(\phi,k)$.
A node $i$ on an augmented path gathers all the discovery data packets during a time window with length $T_i^{I}$ (with unit of timeslot) which starts from transmitting a discovery interest packet $I$ with name $``(\phi,k)\leftarrow\cdots\leftarrow (\phi,0)"$, where node $i$ determines the deadline of accepting the discovery data packets based on $T_i^{I}$. The $T_i^{I}$ value decreases by $1$ timeslot as packet $I$ travels, i.e., $T_j^{I}=T_i^{I}-1$, if packet $I$ flows over link $(i,j)$. At the end of the time window, node $i$ calculates/updates the minimum abstract distance values, and the updates are shown according to Algorithm \ref{alg_discovery}. 

\section{Service Discovery Assisted Dynamic Orchestration (SDADO)}
\label{sec_sdado}
In context of NDN, we propose the SDADO framework that enables each node to make orchestration decisions for the arrived and queued interest packets based on the observed 1-hop-range interest backlogs and the discovered minimum abstract distances. Given Assumption \ref{assume_L_quantitization} and \ref{assume_integer_multiple}, Algorithm \ref{alg_data_producing}-\ref{alg_FwdRateAsgmt} present SDADO at each node $i$ in timeslot $t$.
\begin{assume}
	\label{assume_L_quantitization}
	If $L_i^{(\phi,k)}\neq L_j^{(\phi',k')}$, there exists a constant $ \delta_L$ such that $|L_i^{(\phi,k)}-L_j^{(\phi',k')}|\ge\delta_L$ for all $ (\phi,k), (\phi',k'), i,j$.
\end{assume}
\begin{assume}
	\label{assume_integer_multiple}
	$\bar C_{ji}$, $C_i^{\emph{dp}}$, $C_i^{\emph{pr}}$, $z^{(\phi,k)}$, $r^{(\phi,k)}$ satisfy at least one of the two conditions: i) $\bar C_{ji}$ and $C_i^{\emph{dp}}$ are integer multiples of $z^{(\phi,k)}$, and $C_i^{\emph{pr}}$ is integer multiples of $r^{(\phi,k)}$, for all $(i,j), \phi,k$; ii) $z^{(\phi,k)}=z$ and $r^{(\phi,k)}=r$ for all $\phi,k$.
\end{assume}
Assumption \ref{assume_L_quantitization} is easily satisfied given quantization of abstract distance values, where $\delta_L$ can be the quantization granularity. 
Assumption \ref{assume_integer_multiple} is needed to guarantee the throughput optimality of Algorithm \ref{alg_data_producing}-\ref{alg_FwdRateAsgmt}, where the commodity chosen to allocate resource in each timeslot gets the full capacity. SDADO does not loose throughput-optimality even without Assumption \ref{assume_integer_multiple} but requires \emph{knapsack} operations \cite{bib_Algorithm} to allocate resources (see Section \ref{sec_analysis} for details).
\begin{algorithm}
	\caption{SDADO-data-producing at node $i\in \cup_{\phi} \mathcal P_{(\phi,0)}$ in timeslot $t$}
	\label{alg_data_producing}
	\begin{algorithmic}[1]
		\REQUIRE $U_i^{(\phi,0,c)}(t), \forall (\phi,c)$.
		\ENSURE $p_i^{(\phi,0,c)}(t), \forall (\phi,c)$.
		\FOR {all $(\phi,c)$ such that $i\in \mathcal P_{(\phi,0)}$}
		\STATE Initialize $p_i^{(\phi,0,c)}(t)\gets 0$.
		\ENDFOR
		\STATE Choose $(\phi^*,c^*)\gets \arg\max_{(\phi,c)}\{U_i^{(\phi,0,c)}(t)\}$.\label{line_dp_mwm_begin}
		\IF{$U_i^{(\phi^*,0,c^*)}(t) > 0$}
		\STATE $p_i^{(\phi^*,0,c^*)}(t)\gets \frac{C_i^{\text{dp}}}{z^{(\phi^*,0)}}$.
		\ENDIF\label{line_dp_mwm_end}
	\end{algorithmic}
\end{algorithm}

SDADO-data-producing and -processing-commitment are respectively described in Algorithm \ref{alg_data_producing} and \ref{alg_proc_commitment}, where we define the \emph{virtual backlog} $U_i^{(\phi,k,c)}(t)\triangleq z^{(\phi,k)}Q_i^{(\phi,k,c)}(t)$ as interest commodity $(\phi,k,c)$ backlog weighted by data packet's length. Both Algorithm \ref{alg_data_producing} and \ref{alg_proc_commitment} adopt \emph{max-weight-matching} to allocate data producing rates and processing commitment rates. Particularly in Algorithm \ref{alg_proc_commitment}, 
the operations in line \ref{line_Wpr}, \ref{line_pr_mwm_begin} demonstrate that SDADO allocates the processing capacity to the function with local high demand of input commodity interests and low demand of output commodity interests.
\begin{algorithm}
	\caption{SDADO-processing-commitment at node $i\in \cup_{\phi} \mathcal P_{(\phi,k>0)}$ in timeslot $t$}
	\label{alg_proc_commitment}
	\begin{algorithmic}[1]
		\REQUIRE $U_i^{(\phi,k,c)}(t), \forall (\phi,k>0,c)$
		\ENSURE $p_i^{(\phi,k,c)}(t), \forall (\phi,k>0,c)$
		\FOR {all $(\phi,k,c)$ such that $i\in \mathcal P_{(\phi,k)}$}
		\STATE  $W_{i,\text{pr}}^{(\phi,k,c)}\!(t)\! \gets\! \frac{U_i^{(\phi,k,c)}\!(t)z^{(\phi,k)}-U_i^{(\phi,k-1,c)}\!(t)z^{(\phi,k-1)}}{r^{(\phi,k)}}$.\label{line_Wpr}
		\STATE Initialize $p_i^{(\phi,k,c)}(t)\gets 0$.
		\ENDFOR
		\STATE Choose $(\phi^*,k^*,c^*)\gets \arg\max_{(\phi,k,c)}\{W_{i,\text{pr}}^{(\phi,k,c)}(t)\}$.\label{line_pr_mwm_begin}
		\IF{$W_{i,\text{pr}}^{(\phi,k,c)}(t)>0$}
		\STATE $p_i^{(\phi^*,k^*,c^*)}(t)\gets \frac{C_i^{\text{pr}}}{r^{(\phi^*,k^*)}}$.
		\ENDIF\label{line_pr_mwm_end}
	\end{algorithmic}
\end{algorithm}
\begin{algorithm}
	\caption{Pre-processing for online interest forwarding}
	\label{alg_forward_preparation}
	\begin{algorithmic}[1]
		\REQUIRE ${\bf L}_i\triangleq \{L_v^m: \forall v \in \{i\} \cup \mathcal O(i), \forall m\}$.
		\ENSURE ${\bf J}_{i,m}^{\uparrow}$, $\forall m$; $\Delta L_{ij}^m$, $\mathcal M_{ij}^m \triangleq (\mathcal M_{ij,m}^{>}, \mathcal M_{ij,m}^{=}, \mathcal M_{ij,m}^{<})$, $\forall m$, $\forall j\in \mathcal O(i)$.
		\FOR{ $\forall m$}
		\FOR{$\forall j\in \mathcal O(i)$}
		\STATE Calculate $\tilde L_{i,j}^{m} \!\gets\! l_{ij} + L_j^{m}$; $\Delta L_{ij}^{m} \!\gets\! L_i^{m}\!-\!L_j^{m}$.\label{line_delta_L}
		\ENDFOR
		\STATE Sort $\{\tilde L_{i,j}^{m}: \forall j\in \mathcal O(i)\}$ to $\tilde L_{i,j_1}^{m} \le \cdots \le \tilde L_{i,j_{|\mathcal O(i)|}}^{m}$, and obtain sequence ${\bf J}_{i,m}^{\uparrow}\gets [j_1,\cdots,j_{|\mathcal O(i)|}]$. \label{line_sort_deltaL_amg_links}
		\ENDFOR
		\FOR{$\forall j\in \mathcal O(i)$}
		\STATE Sort $\{\Delta L_{ij}^{m}: \forall m\}$ to $\Delta L_{ij}^{m_1} \ge \cdots \ge \Delta L_{ij}^{m_N}$, and obtain sequence ${\bf M}_{ij}^{\downarrow}\gets [m_1,\cdots, m_N]$.\label{line_sort_deltaL_amg_commo}
		\FOR{$n=1$ to $N$}\label{line_generate_M_sets_begin}
		\STATE $\hat m\triangleq {\bf M}_{ij}^{\downarrow}[n]$;
		\STATE Find $\mathcal M_{ij,\hat m}^{>}\gets \{m: \Delta L_{ij}^{m} > \Delta L_{ij}^{\hat m}\}$; \\$\mathcal M_{ij,\hat m}^{=}\gets \{m: \Delta L_{ij}^{m} = \Delta L_{ij}^{\hat m}\}$; \\$\mathcal M_{ij,\hat m}^{<}\gets \{m: \Delta L_{ij}^{m} < \Delta L_{ij}^{\hat m}\}$.
		\ENDFOR \label{line_generate_M_sets_end}
		\ENDFOR
	\end{algorithmic}
\end{algorithm}

Algorithm \ref{alg_forward_preparation}-\ref{alg_FwdRateAsgmt} present SDADO-interest-forwarding based on ${\bf U}_i(t)$ and ${\bf L}_i$. We denote $m\!\triangleq\! (\phi,k,c)$ and $N\!\triangleq\!|\{(\phi,k,c)\}|$ for convenience.
Algorithm \ref{alg_forward_preparation} describes the pre-processing on ${\bf L}_i$ before SDADO online operations. Algorithm \ref{alg_fowarding} describes the SDADO online interest forwarding: node $i$ prioritizes the commodities by sorting metrics $\Theta_i^m(t)$ (line \ref{line_sort_Theta}), which intuitively implies giving higher forwarding priority to interest commodity with higher delay-reduction gain weighted by congestion-reduction gain, and then schedules commodities for forwarding according to the priority order, where $AsgmtInd_{ij}(t)$ indicates whether link $(i,j)$ has been scheduled. Algorithm \ref{alg_FwdRateAsgmt} is called in line \ref{line_call_asgmt} of Algorithm \ref{alg_fowarding} judging whether to forward a commodity $\tilde m'$ over link $(i,j')$. 

\begin{algorithm}
	\caption{SDADO-interest-forwarding at node $i\in \mathcal V$ in timeslot $t$}
	\label{alg_fowarding}
	\begin{algorithmic}[1]
		\REQUIRE ${\bf U}_i(t)\triangleq \{U_v^m(t): \forall v \!\in\! \{i\} \cup \mathcal O(i), \forall m\}$; 		
		${\bf J}_{i,m}^{\uparrow}$, $\forall m$; $\Delta L_{ij}^m$, $\mathcal M_{ij}^m$, $\forall m$, $\forall j\in \mathcal O(i)$; $h_i\triangleq 2\sum_{j\in \mathcal O(i)}{\bar C_{ji}}$.
		\ENSURE $b_{ij}^{m}(t)$, $\forall m$, $\forall j\in \mathcal O(i)$
		\FOR{$\forall m$}
		\FOR{$\forall j\in \mathcal O(i)$}
		\STATE Calculate $\Delta U_{ij}^{m}(t) \gets U_i^{m}(t)-U_j^{m}(t)$.\label{line_diff_U}
		\STATE Initialize $b_{ij}^{m}(t)\gets 0$;
		\ENDFOR
		\STATE Calculate $\Theta_i^m(t)\gets \sum\nolimits_{j\in \mathcal O(i)}{\left[\Delta U_{ij}^{m}(t)\right]^+\left[\Delta L_{ij}^{m}(t)\right]^+}$.\label{line_Theta}
		\ENDFOR
		\STATE Sort $\{\Theta_i^m(t): \forall m\}$ to $\Theta_i^{\tilde m_1}(t)\ge\cdots\ge \Theta_i^{\tilde m_M}(t)$ and obtain sequence 
		${\tilde {\bf M}}_{i}^{\downarrow} \gets  [\tilde m_1,\cdots,\tilde m_M]$.\label{line_sort_Theta}
		\STATE Denote $\Delta {\bf U}_{ij}(t)\!\!\triangleq\!\! \{\Delta U_{ij}^m(t):\forall m\}$, $\Delta {\bf L}_{ij}\!\!\triangleq\!\! \{\Delta L_{ij}^m:\forall m\}$, for all $j\in \mathcal O(i)$.
		\STATE Initialize $AsgmtInd_{ij}(t)\gets\FALSE$ for all $j\in \mathcal O(i)$.\label{line_asgmt_begin}
		\FOR{$n=1$ to $N$}
		\STATE $\tilde m' \triangleq {\tilde {\bf M}}_{i}^{\downarrow}[n]$
		\FOR{$s=1$ to $|\mathcal O(i)|$}
		\STATE $j' \triangleq {\bf J}_{i,\tilde m'}^{\uparrow}[s]$;
		\IF{$AsgmtInd_{ij'}(t)=\FALSE$}
		\STATE FwdRateAsgmt$(\Delta {\bf U}_{ij'}(t), \Delta {\bf L}_{ij'}, \mathcal M_{ij'}^{\tilde m'},h_i)$.\label{line_call_asgmt}
		\ENDIF
		\ENDFOR
		\ENDFOR\label{line_asgmt_end}
	\end{algorithmic}
\end{algorithm}
\begin{algorithm}
	\caption{FwdRateAsgmt}
	\label{alg_FwdRateAsgmt}
	\begin{algorithmic}[1]
		\REQUIRE $\Delta {\bf U}_{ij'}(t)$, $\Delta {\bf L}_{ij'}$, $\mathcal M_{ij'}^{\tilde m'}$, $h_i$.
		\ENSURE $b_{ij'}^{\tilde m'}(t)$, $AsgmtInd_{ij'}(t)$
		\IF {$ \nexists m\in \mathcal M_{ij',\tilde m'}^{>}$ satisfying $\Delta U_{ij'}^{m}(t) \ge \Delta U_{ij'}^{\tilde m'}(t)$ 
			\AND $\nexists m\in \mathcal M_{ij',\tilde m'}^{=}$ satisfying $\Delta U_{ij'}^{m}(t) > \Delta U_{ij'}^{\tilde m'}(t)$
		\AND $\nexists m\in \mathcal M_{ij',\tilde m'}^{<}$ satisfying $\Delta U_{ij'}^{m}(t) > \Delta U_{ij'}^{\tilde m'}(t) + h_i$
		\AND $\mathbbm{1}(\Delta L_{ij'}^{\tilde m'}\!\!>\!\! 0)\mathbbm{1}(\Delta U_{ij'}^{\tilde m'}\!(t) \!\!\ge\!\! -h_i) \!+\! \mathbbm{1}(\Delta L_{ij'}^{\tilde m'}\!\le\! 0)\mathbbm{1}(\Delta U_{ij'}^{\tilde m'}(t)\ge h_i)>0$}\label{line_4conditions}
	\FOR{$\forall m \in \mathcal M_{ij',\tilde m'}^{>} \cup \mathcal M_{ij',\tilde m'}^{<}$}\label{line_rate_assign_begin}
	\STATE Calculate 
	$\omega_{ij'}^{m,\tilde m'}(t)\triangleq \frac{\Delta U_{ij'}^{m}(t)-\Delta U_{ij'}^{\tilde m'}(t)}{\Delta L_{ij'}^{\tilde m'}-\Delta L_{ij'}^m}$.\label{line_omega}
	\ENDFOR
	\STATE Calculate ${ B_{ij',\tilde m'}^{\text{lb}}(t)\!\!\gets\!\! \left[\max\nolimits_{m\in \mathcal M_{ij',{\tilde m'}}^{<}}\{\omega_{ij'}^{m,\tilde m'}\!(t)\}\!\right]^{\!+}}$;\label{line_B_hat}
	${ B_{ij'\!,\tilde m'}^{\text{ub}}\!(t)\gets \min\nolimits_{m\in \mathcal M_{ij'\!,{\tilde m'}}^{>}}\!\{\omega_{ij'}^{m,\tilde m'}\!\!(t)\}}$.\label{line_bounds_begin}
	\IF{$\Delta L_{ij'}^{\tilde m'}>0$}
	\STATE Update $B_{ij'\!,\tilde m'}^{\text{lb}}(t)\gets \max\{\! B_{ij'\!,\tilde m'}^{\text{lb}}(t),\frac{-\Delta U_{ij'}^{\tilde m'}\!(t)}{\Delta L_{ij'}^{\tilde m'}}\}$.\label{line_B_1}
	\ELSIF{$\Delta L_{ij'}^{\tilde m'}<0$}
	\STATE Update ${ B_{ij'\!,\tilde m'}^{\text{ub}}\!(t)\gets \min\{ B_{ij'\!,\tilde m'}^{\text{ub}}\!(t),\frac{-\Delta U_{ij'}^{\tilde m'}\!(t)}{\Delta L_{ij'}^{\tilde m'}}\}}$.\label{line_B_2}
	\ENDIF\label{line_bounds_end}			
	\IF {$B_{ij',\tilde m'}^{\text{ub}} \ge B_{ij',\tilde m'}^{\text{lb}}$}\label{line_rate_assign_begin2}
	\STATE $b_{ij'}^{\tilde m'}(t)\gets \nicefrac{\bar C_{j'i}}{z^{\tilde m'}}$.\label{line_bij}
	\STATE $AsgmtInd_{ij'}(t)\gets \TRUE$.\label{line_asgmtInd}
	\ENDIF\label{line_rate_assign_end}
	\ENDIF
\end{algorithmic}
\end{algorithm}

As is shown in Algorithm \ref{alg_forward_preparation}, node $i$ pre-processes ${\bf L}_i$ to prepare for making the forwarding decisions by calculating the differential abstract distances $\Delta L_{ij}^m$ (line \ref{line_delta_L}), and then sorting all the $\Delta L_{ij}^m$ among all the outgoing links for a single commodity (line \ref{line_sort_deltaL_amg_links}) and among all the commodities over a single link (line \ref{line_sort_deltaL_amg_commo}). 
This pre-processing only needs to run once before the online implementation of SDADO in stationary networks or repeat infrequently in networks with time-varying topology. 

Algorithm \ref{alg_fowarding} describes the SDADO online forwarding of interest packets at node $i$ in timeslot $t$. Node $i$ first prioritizes the commodities to be forwarded by the sorting metric values $\Theta_i^m(t)$ (line \ref{line_sort_Theta}), which is calculated based on $\Delta U_{ij}^m(t)$ and $\Delta L_{ij}^{m}$ . The intuition behind $\Theta_i^m(t)$ is reflected in its formulation (line \ref{line_Theta}). SDADO gives higher priority to forwarding the interest commodity which has higher total delay reduction gain characterized by $\Delta L_{ij}^{m}$, weighted by higher traffic reduction urgency characterized by the differential virtual backlog $\Delta U_{ij}^{m}(t)$, over all the outgoing links. The priority determines scheduling order of forwarding rate assignment among commodities. The scheduling is described in line \ref{line_asgmt_begin}-\ref{line_asgmt_end}, where $AsgmtInd_{ij}(t)$ indicates whether link $(i,j)$ has been scheduled for forwarding. The complexity of SDADO on interest forwarding over a link per timeslot is $O(N^2)$ in the worst case. 

Algorithm \ref{alg_FwdRateAsgmt} describes ``FwdRateAsgmt" (called in line \ref{line_call_asgmt} of Algorithm \ref{alg_fowarding}) judging whether to allocate the average reverse link capacity of a outgoing link $(i,j')$ to a targeted commodity $\tilde m'$. Line \ref{line_4conditions} shows four conditions that commodity $\tilde m'$ has to satisfy in order to be scheduled. The operations in line \ref{line_rate_assign_begin}-\ref{line_rate_assign_end} are used to further determine the eligibility of allocating forwarding resource to interest commodity $\tilde m'$ over link $(i,j')$ to maintain network stability, which involves comparing two bounding values $B_{ij',\tilde m'}^{\text{ub}}$ and $B_{ij',\tilde m'}^{\text{lb}}$ calculated based on the metrics $\omega_{ij'}^{m,\tilde m'}(t)$ (line \ref{line_omega} and \ref{line_bounds_begin}-\ref{line_bounds_end}). As will be shown in Section \ref{sec_analysis}, scheduling the commodity that satisfy the four conditions and the condition $B_{ij',\tilde m'}^{\text{ub}} \ge B_{ij',\tilde m'}^{\text{lb}}$ in each timeslot guarantees the network stability.

\section{Algorithm Analysis}
\label{sec_analysis}
In this section, we analyze throughput-optimality of SDADO and then further describe the SDADO algorithm without Assumption \ref{assume_integer_multiple}. We start with the following lemma. 
\begin{lem}
	\label{lemma_sdado_routing}
	With Assumption \ref{assume_L_quantitization} and \ref{assume_integer_multiple}, there exists a coefficient $\eta_{ij}(t)$ satisfying $0\le \eta_{ij}(t)\le \nicefrac{h_i}{\delta_L}$, $\forall (i,j),t$, such that the optimal solution ${\bf b}_{ij}^*(t)\triangleq \{{b_{ij}^{m}}^*(t)\}$ of the problem 
	\begin{subequations}
		\label{eq_sdado-routing_opt}
		\begin{align}
				&\text{max} &&
				\Scale[1]{\sum\nolimits_{m}{\left[\Delta U_{ij}^m(t)+\eta_{ij}(t)\Delta L_{ij}^m\right]z^m b_{ij}^m(t)},}    \label{eq_sdado_routing_obj}\\
				&\text{s.t.} && \Scale[1]{\sum\nolimits_{m}{z^mb_{ij}^m(t)} \le \bar C_{ji},\ \ b_{ij}^m(t)\in \mathbb{Z}^+,}   \label{eq_sdado_routing_cap_constraint}
			\end{align}
	\end{subequations} 
	and the SDADO-interest-forwarding solution ${\bf b}_{ij}(t)$ obtained in Algorithm \ref{alg_forward_preparation}-\ref{alg_FwdRateAsgmt} satisfy $Z_{\eta_{ij}(t)}({\bf b}_{ij}(t))\ge Z_{\eta_{ij}(t)}({\bf b}_{ij}^*(t)) - [\nicefrac{l_{ij}}{\delta_L}-1]h_i \bar C_{ji}$, where $Z_{\eta_{ij}(t)}(*)$ represents the objective function \eqref{eq_sdado_routing_obj} given parameter $\eta_{ij}(t)$.
\end{lem}
\begin{proof}
	The proof is shown in Appendix \ref{appendix_lemma_proof}.
\end{proof}


Let $ {\bm \chi}$ represent a vector with the same length as ${\bf U}(t)$ having the $\lambda_c^\phi z^{(\phi,K_\phi)}$ as the $(c,(\phi,K_\phi,c))$th elements and $0$ elsewhere; let $\Lambda$ represent the \emph{computing network capacity region} \cite{bib_DCNC} defined as the closure of all ${\bm \chi}$ that can be stabilized by the computing network. Then based on Lemma \ref{lemma_sdado_routing}, we can further show that SDADO is throughput optimal:
\begin{thm}
	\label{thm_stability}
	With Assumption \ref{assume_L_quantitization} and \ref{assume_integer_multiple}, for any ${\bm \chi}$ in the interior of $\Lambda$, i.e., $\exists \epsilon > 0$ such that ${(\bm \chi + \epsilon \bf 1)}\in \Lambda$, SDADO in Algorithm \ref{alg_data_producing}-\ref{alg_FwdRateAsgmt} stabilizes the network, i.e.,
	\begin{equation}
		\label{eq_stability}
		{\limsup\limits_{t\rightarrow \infty} \frac{1}{t}\sum\nolimits_{\tau, m, i} U_i^m(\tau) \le \frac{B_0 + B_1}{\epsilon},}\ \ \text{with prob. 1},
	\end{equation}
	where $B_0$ and $B_1$ are constants depending on the network parameters $\bar C_{ij}$, $C_i^{\emph{pr}}$, $C_i^{\emph{dp}}$, $z^{m}$, $r^{m}$, $\delta_L$, $l_{ij}$.
\end{thm}
\begin{proof}
	The proof is shown in Appendix \ref{appendix_thm_proof}.
\end{proof}

In fact, Assumption \ref{assume_integer_multiple} can be removed without affecting the throughput-optimality of SDADO, but we need to make modifications on Algorithm \ref{alg_data_producing}, \ref{alg_proc_commitment}, and \ref{alg_FwdRateAsgmt}: i) 
the max-weight-matching in line \ref{line_dp_mwm_begin}-\ref{line_dp_mwm_end} of Algorithm \ref{alg_data_producing} and line \ref{line_pr_mwm_begin}-\ref{line_pr_mwm_end} of Algorithm \ref{alg_proc_commitment} can be replaced by knapsack; ii) 
instead of implementing line \ref{line_bij} in Algorithm \ref{alg_FwdRateAsgmt}, we calculate $\eta_{ij'}(t)=\min\{B_{ij',\tilde m'}^{\text{lb}}+1, \frac{1}{2}[B_{ij',\tilde m'}^{\text{lb}} + B_{ij',\tilde m'}^{\text{ub}}]\}$ and plug it into \eqref{eq_sdado_routing_obj}, and then solve \eqref{eq_sdado-routing_opt} by knapsack. Thus, we further have the following corollary.

\begin{coro}
	With Assumption \ref{assume_L_quantitization}, for any $ {\bm \chi}$ in the interior of $\Lambda$, the SDADO algorithm with knapsacks stabilizes the network.
\end{coro}


\section{Numerical Results}
\label{sec_simulation}
In this section, we evaluate performance of the SDADO scheme and compare it with state-of-art: the DCNC and EDCNC schemes \cite{bib_DCNC} and the Best Route scheme \cite{bib_nfd_guider} extended to NDN-based computing networks. Fig. \ref{fig_fog_network} shows the simulated network with Fog topology. The network and service settings are listed in Table \ref{tab_fog_net}.
\begin{table}[ht]
	\begin{center}
		\caption{Network and Service Settings}
		\label{tab_fog_net}
		\begin{tabular}{c||c}
			\toprule
			 \!$C_i^{\text{pr}}\!=\!2\!\times\! 10^7 \text{cycles/sec.}, 1\!\le\!i\!\le\!4$; & $ C_{ij}\!=\!800 \text{Mbps}, 1\!\le\!(i,j)\!\le\!5 $;\!\\
			 \!$C_i^{\text{pr}}\!=\!5\!\times\! 10^6 \text{cycles/sec.}, 5\!\le\!i\!\le\!9$. & $ C_{ij}\!=\!640 \text{Mbps}, 6\!\le\!(i,j)\!\le\!14 $.\!\\
			\hline
			\multicolumn{2}{c}{$C_i^{\text{dp}}$ is sufficiently large for $i=10,\cdots,13$ and $16,\cdots,19$.}\\
		\end{tabular}
		\begin{tabular}{c||c|c}
			\hline \hline
			Wireless Link Index & 15-24 & 25-29\\
			\hline
			Small Scale Fading & Rician, Rician-factor$=10$dB. & Rayleigh. \\
			\hline
			Min. Avg. Rx-SNR:& \multicolumn{2}{c}{ $\bullet$ 5dB given Edge-Tx; $\bullet$ 0dB given Mesh-Tx.}\\
			\hline
			\multirow{2}{*}{Other Settings} & \multicolumn{2}{c}{$\bullet$ $10$MHz bandwidth; $\bullet$ Free space pathloss;}\\
				 & \multicolumn{2}{c}{$\bullet$ Log-normal shadowing.} \\
		\end{tabular}
		\begin{tabular}{c||c|c}
			\hline \hline
			Service Functions & $(1,0), (1,1), (1,2)$. & $(2,0), (2,1), (2,2)$. \\
			\hline
			$\mathcal C_\phi$ &  10, 12, 14. & 15, 17, 19. \\
			\hline
			$\mathcal P_{\phi,0}$ & 16-19 & 10-13 \\
			\hline
			\multirow{2}{*}{$z^{(\phi,k)}$ (KB)} & \!$z^{(1,0)}\!=\!z^{(1,1)}\!=\!50$;\! & \!$z^{(2,0)}\!=\!z^{(2,1)}\!=\!50$;\!\\
			& $z^{(1,2)}=100$. & $z^{(2,2)}=20$. \\
			\hline
			\multirow{2}{*}{$r^{(\phi,k)}$ (cycles/packet)} & $r^{(1,1)}=5\times10^3$; & $r^{(2,1)}=5\times10^3$; \\
			& $r^{(1,2)}=2\times10^4$. & $r^{(2,2)}=10^4$.\\
			\bottomrule
		\end{tabular}
	\end{center}
	\vspace{-0.3cm}
\end{table}

\begin{figure*}[ht]
	\centering
	\hspace{-0.5cm}
	\subfigure[]{
		\centering \includegraphics[width=4.4cm,]{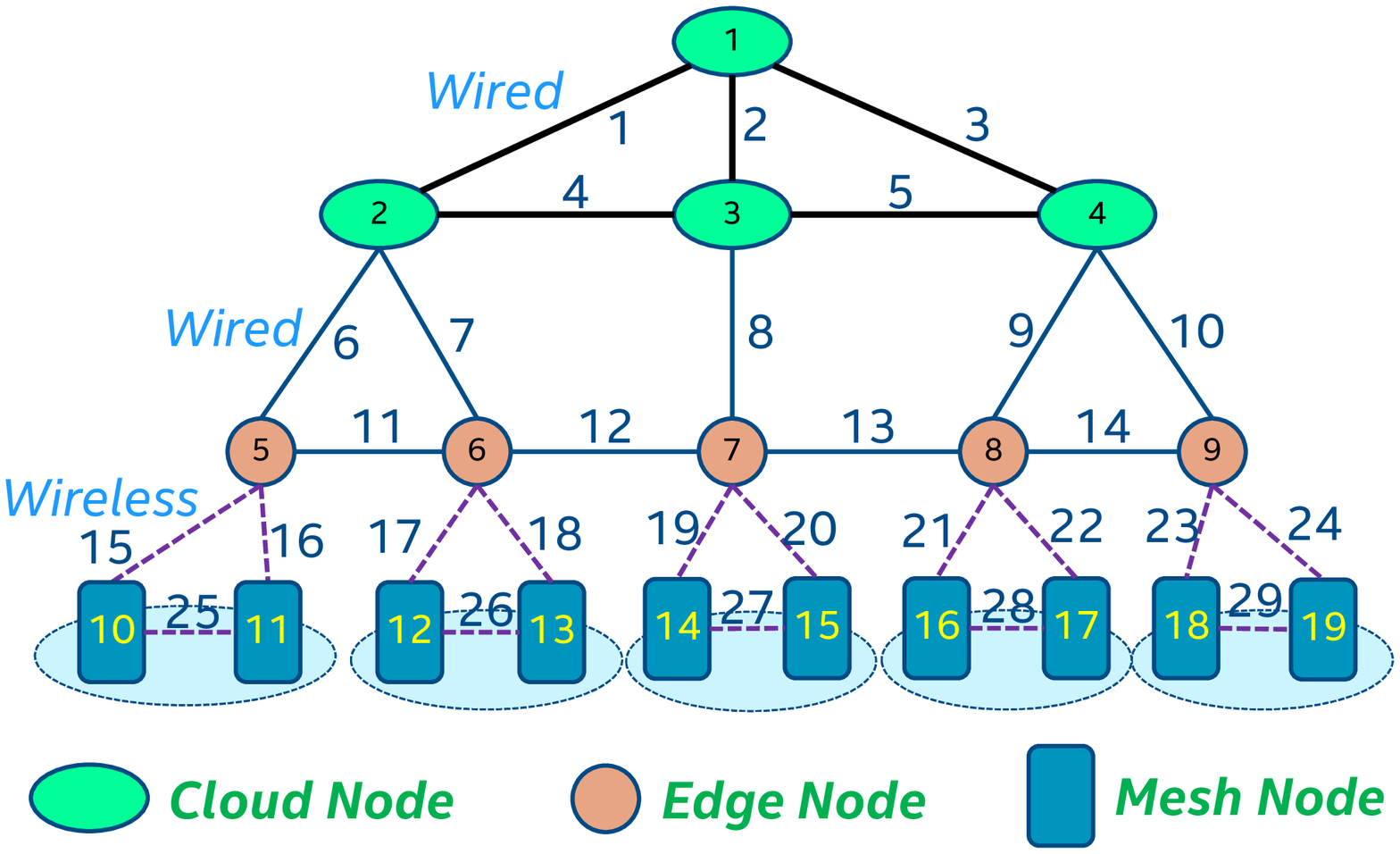}
		\label{fig_fog_network}
	}
	\hspace{-0.5cm}
	\subfigure[]{
		\centering \includegraphics[width=4.9cm,]{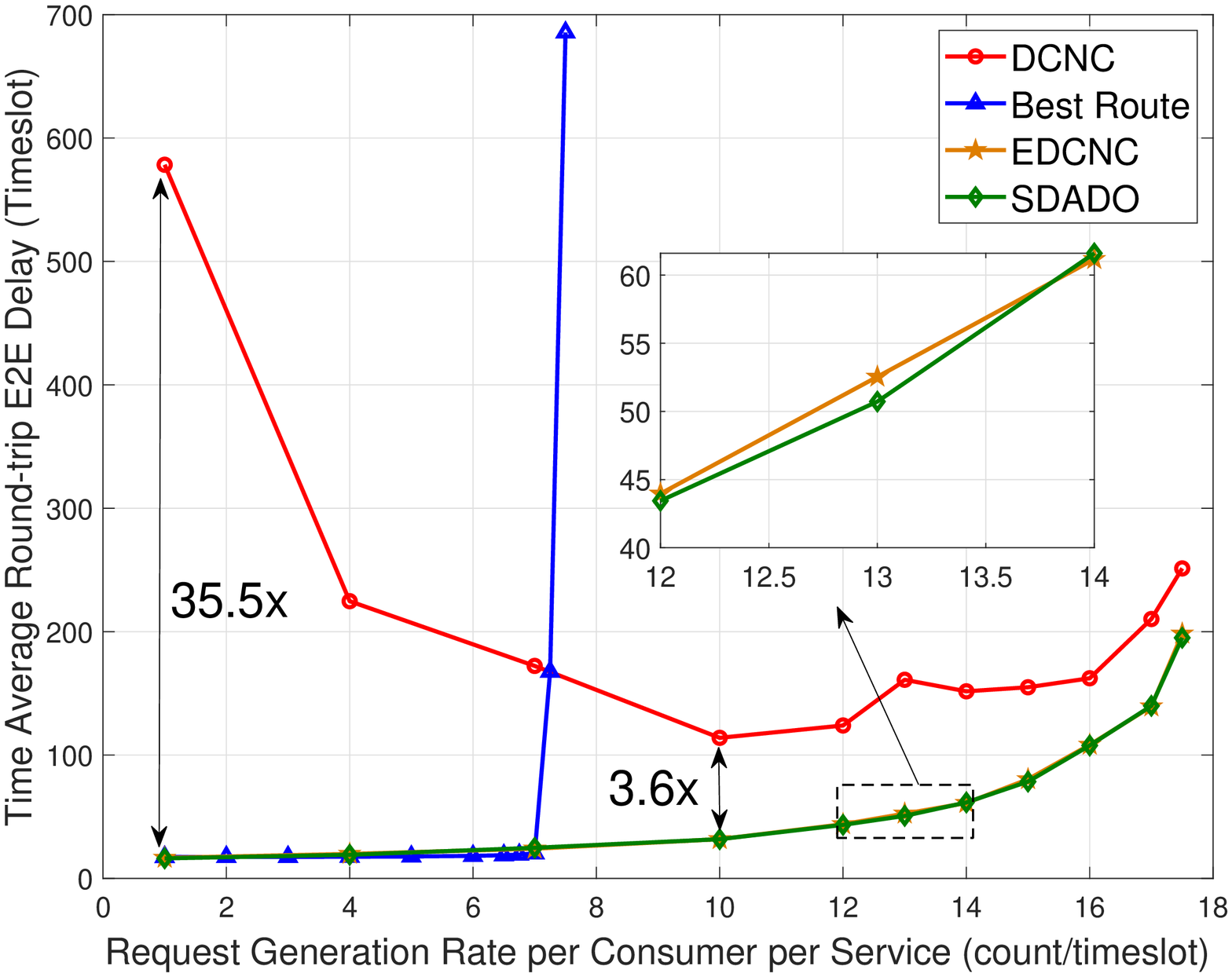}
		\label{fig_delay}
	}
	\hspace{-0.9cm}
	\subfigure[]{
		\centering \includegraphics[width=4.9cm]{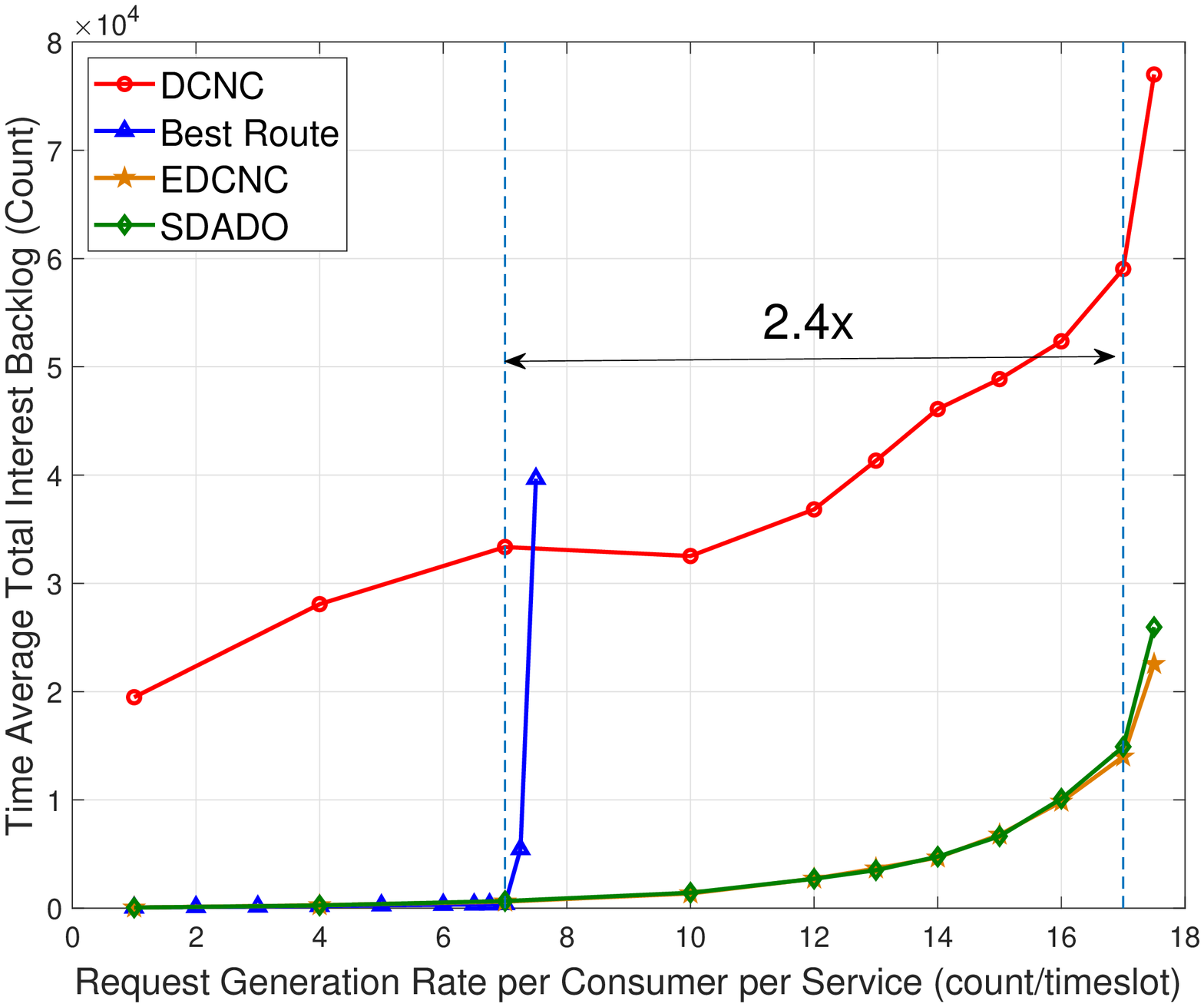}
		\label{fig_interest_backlog}
	}
	\hspace{-0.85cm}
	\subfigure[]{
		\centering \includegraphics[width=4.9cm]{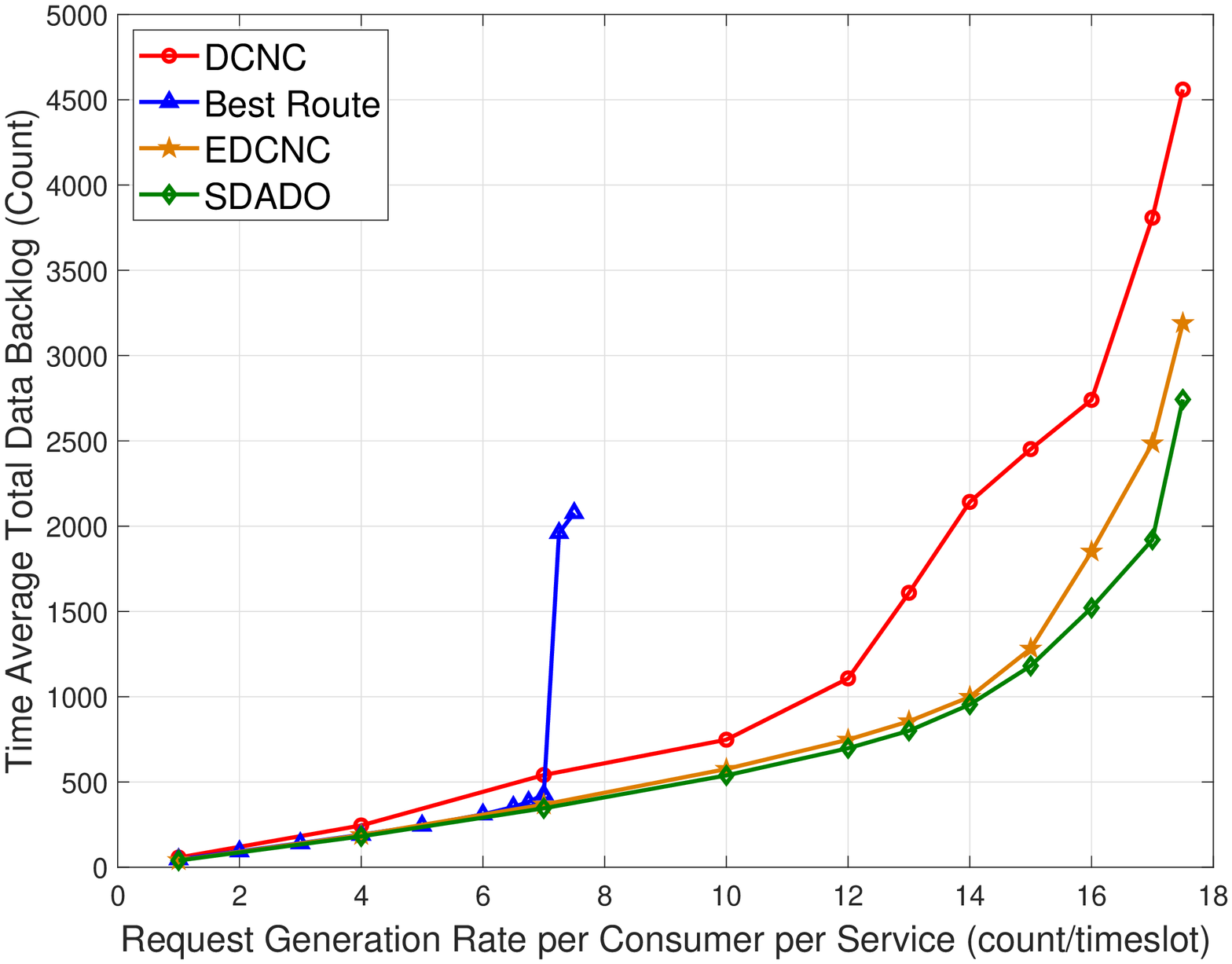}
		\label{fig_data_backlog}
	}
	\hspace{-0.85cm}
	\caption{Simulation results. a)~Simulation network with Fog topology; b)~round trip E2E delay; c)~time average total interest backlog; d)~time average total data backlog.} 
	\vspace{-0.5cm}
	\label{fig_simulation results}
\end{figure*}

Fig. \ref{fig_delay} shows the average round trip E2E delay performances of DCNC, Best Route, EDCNC, and SDADO. In low traffic scenario (request rate $\lambda_c^\phi \le 7.0$ count/timeslot), Best Route, EDCNC, and SDADO achieve much lower E2E delay than DCNC. This is because the former three algorithms enable topological information to dominate routing and guide the interest packets flowing along short paths, while DCNC has to explore paths by sending interest packets in different directions, which sacrifices delay performance. As $\lambda_c^\phi$ exceeds $7.0$ count/timeslot, the E2E delay of Best Route sharply increases to unbearable magnitude indicating that its throughput limit is exceeded. In contrast, DCNC, EDCNC, SDADO still support the high traffic due to their flexibility, while DCNC and EDCNC outperform DCNC by achieving lower E2E delay ($3.6\times$ -- $35.5\times$ gain). Although Fig. \ref{fig_delay} shows that SDADO and EDCNC have the similar E2E delay performance, SDADO is much easier to implement in distributed manner than EDCNC. To implement EDCNC, one has to heuristically find a global bias coefficient $\eta$ ($=\!10^8$ in the simulation) and distribute it among all the nodes. Moreover, given an arbitrary network, there is no clear clue of systematically finding a proper value of $\eta$, even it significantly influences the E2E delay. In contrast, SDADO does not use any fixed global bias coefficient but dynamically hybrid queuing-based flexibility and topology-based discipline in orchestration.   

Throughput performances are further demonstrated in Fig. \ref{fig_interest_backlog} by showing the time average interest backlog accumulations evolving with traffics. When $\lambda_c^\phi > 7.0$ count/timeslot, the average total interest backlog under Best Route exhibits a sharp increase, illustrating its throughput limit. In comparison, the sharp backlog increase under DCNC, EDCNC, and SDADO occurs at $\lambda_c^\phi=17.0$ count/timeslot exhibiting $2.4\times$ throughput gain. In addition, throughput limits of DCNC, EDCNC, and SDADO are the same because all of them are throughput optimal. In the meanwhile, backlog accumulation levels of SDADO and EDCNC are much lower than DCNC indicating significantly reduced E2E delay. The time average accumulated data backlogs are shown in Fig. \ref{fig_data_backlog}, which verifies the similar trend.

\section{Conclusions}
\label{sec_conclusion}
In this paper, we consider dynamic orchestration problem in NDN-based distributed computing networks to deliver next-generation services having function-chaining structures. We propose the SDADO algorithm that reduces E2E delay while achieving throughput-optimality by dynamically hybriding queuing-based flexibility and topology-based discipline, where the topological information is obtained via our proposed service discovery mechanism. Theoretical analysis and simulation results have confirmed our conclusions.

\appendices
\section{Proof of Lemma \ref{lemma_sdado_routing}}
\label{appendix_lemma_proof}
With purpose of notational convenience, we first relist the conditions in line \ref{line_4conditions} of Algorithm \ref{alg_FwdRateAsgmt} for a targeted commodity $m$ over link $(i,j)$:
\begin{enumerate}
	\item $\nexists m'\in \mathcal M_{ij,m}^{>}$ satisfies $\Delta U_{ij}^{m'}(t) \ge \Delta U_{ij}^{m}(t)$;\label{cond_1}
	\item $\nexists m'\in \mathcal M_{ij,m}^{=}$ satisfies $\Delta U_{ij}^{m'}(t) > \Delta U_{ij}^{m}(t)$;\label{cond_2}
	\item $\nexists m'\in \mathcal M_{ij,m}^{<}$ satisfies $\Delta U_{ij}^{m'}(t) > \Delta U_{ij}^{m}(t) + h_i$;\label{cond_3}
	\item $\mathbbm{1}(\Delta L_{ij}^{m}\!>\! 0)\mathbbm{1}(\Delta U_{ij}^{m}(t) \!\ge\! -h_i) \!+\! \mathbbm{1}(\Delta L_{ij}^{m}\!\le\! 0)\times\mathbbm{1}(\Delta U_{ij}^{m}(t)\ge h_i)>0$.\label{cond_4}
\end{enumerate}
For commodity $m$ and link $(i,j)$, we further define
\begin{align}
	&\Scale[1]{\hat B_{ij,m}^{\text{lb}}(t) = \left[\max\nolimits_{m'\in \mathcal M_{ij,{m}}^{<}}\left\{\omega_{ij}^{m',m}(t)\right\}\right]^{+},}\label{eq_hat_B_lb}\\
	&\Scale[1]{\hat B_{ij,m}^{\text{ub}}(t) = \max\nolimits_{m'\in \mathcal M_{ij,{m}}^{>}}\left\{\omega_{ij}^{m',m}(t)\right\},}\label{eq_hat_B_ub}
\end{align}
where $\omega_{ij}^{m',m}(t)$ is defined in line \ref{line_omega} of Algorithm \ref{alg_FwdRateAsgmt}. With the definitions of $B_{ij,m}^{\text{lb}}(t)$ and $B_{ij,m}^{\text{ub}}(t)$ in line \ref{line_B_hat}-\ref{line_bounds_end} of Algorithm \ref{alg_FwdRateAsgmt}, it follows that
\begin{align}
	\Scale[1]{B_{ij,m}^{\text{lb}}(t) \!=\!
	\begin{cases}
		\max\left\{\! \hat B_{ij,m}^{\text{lb}}(t),\frac{-\Delta U_{ij}^{m}(t)}{\Delta L_{ij}^{m}}\!\right\}, & \text{if $\Delta L_{ij}^{m}>0$},\\
		\hat B_{ij,m}^{\text{lb}}(t), & \text{otherwise},
	\end{cases}}\label{eq_B_ij_m}\\
	\Scale[1]{B_{ij,m}^{\text{ub}}(t) \!=\!
		\begin{cases}
			\min\left\{\! \hat B_{ij,m}^{\text{ub}}(t),\frac{-\Delta U_{ij}^{m}(t)}{\Delta L_{ij}^{m}}\!\right\}, & \text{if $\Delta L_{ij}^{m}<0$},\\
			\hat B_{ij,m}^{\text{ub}}(t), & \text{otherwise},
	\end{cases}}
\end{align}

With Assumptions \ref{assume_L_quantitization} and \ref{assume_integer_multiple}, when running SDADO-interest-forwarding according to Algorithm \ref{alg_fowarding}-\ref{alg_FwdRateAsgmt} over link $(i,j)$ in timeslot $t$, one of the following two cases should be satisfied:
\begin{enumerate}[i)]
	\item a commodity $m$ is assigned to be forwarded over link $(i,j)$ using full capacity, i.e., $b_{ij}^{m}(t)= \nicefrac{\bar C_{ji}}{z^{m}}$, and commodity $m$ satisfies Condition \ref{cond_1})-\ref{cond_4}) and $B_{ij,m}^{\text{lb}} \le B_{ij,m}^{\text{ub}}$; \label{cond_i}
	\item otherwise, no commodity is assigned to be forwarded over link $(i,j)$, i.e., $b_{ij}^{m'}(t)=0$ for $\forall m'$, but there exists a commodity $m$ satisfying Condition \ref{cond_1})-\ref{cond_3}) and $\hat B_{ij,m}^{\text{lb}} \le \hat B_{ij,m}^{\text{ub}}$. \label{cond_ii}
\end{enumerate}

We then set the value of $\eta_{ij}(t)$ as follows.
\begin{itemize}
	\item If Condition \ref{cond_i}) is satisfied, we set
	\begin{equation}
		\Scale[1]{\eta_{ij}(t)=\min\left\{B_{ij,m}^{\text{lb}}+\kappa, \frac{1}{2}\left[B_{ij,m}^{\text{lb}} + B_{ij,m}^{\text{ub}}\right]\right\},}\label{eq_h_1}
	\end{equation}
	where $\kappa$ is a constant set to be any non-negative value;
	\item if Condition \ref{cond_ii}) is satisfied, we set
	\begin{equation}
		\Scale[1]{\eta_{ij}(t)=\min\left\{\hat B_{ij,m}^{\text{lb}}+\kappa, \frac{1}{2}\left[\hat B_{ij,m}^{\text{lb}} + \hat B_{ij,m}^{\text{ub}}\right]\right\}.}\label{eq_h_2}
	\end{equation}
\end{itemize}

For commodity $m$ satisfying Condition \ref{cond_3}), we have
\begin{align}
	&\Scale[1]{\omega_{ij}^{m', m}(t)= \frac{\Delta U_{ij}^{m'}(t)-\Delta U_{ij}^{m}(t)}{\Delta L_{ij}^{m}-\Delta L_{ij}^{m'}}\le\frac{h_i}{\Delta L_{ij}^{m}-\Delta L_{ij}^{m'}}
	\le \frac{h_i}{\delta_L},}\notag\\
	&\ \ \ \ \ \ \ \ \ \ \ \ \ \ \ \ \ \ \ \ \ \ \ \ \ \ \ \ \ \ \ \ \ \ \ \ \ \ \ \ \ \ \ \ \ \ \ \Scale[1]{\forall m'\in \mathcal M_{ij,m}^{<};}
\end{align}
which can be plugged into \eqref{eq_hat_B_lb}, and it follows that
\begin{equation}
	\Scale[1]{\hat B_{ij,m}^{\text{lb}}(t)  \le \frac{h_i}{\delta_L}}.\label{eq_bounding_1}\textbf{}
\end{equation}
If commodity $m$ also satisfies Condition \ref{cond_4}), we have
\begin{equation}
	\Scale[1]{\frac{-\Delta U_{ij}^m(t)}{\Delta L_{ij}^m} \le \frac{-\Delta U_{ij}^m(t)}{\delta_L} \le \frac{h_i}{\delta_L}},\ \text{if $\Delta L_{ij}^m>0$},\label{eq_bounding_2}
\end{equation}
and by plugging \eqref{eq_bounding_1} and \eqref{eq_bounding_2} into \eqref{eq_B_ij_m}, we further obtain
\begin{equation}
	\Scale[1]{B_{ij,m}^{\text{lb}}(t) \le \frac{h_i}{\delta_L}.}\label{eq_bounding_3}
\end{equation}
After plugging \eqref{eq_bounding_1} and \eqref{eq_bounding_3} respectively into \eqref{eq_h_2} and \eqref{eq_h_1}, we show that $\eta_{ij}(t)$ is upper bounded:
\begin{equation}
	\Scale[1]{h_{ij}(t) \le \frac{h_i}{\delta_L}+\kappa.}\label{eq_h_ij_bound}
\end{equation}

%

Define $W_{ij}(x)\triangleq \max\nolimits_{m'}\{\Delta U_{ij}^{m'}(t) + x\Delta L_{ij}^{m'}\}$, which is a piece-wise linear function of $x$, over $\mathbb{R}^+$. Based on the definitions in \eqref{eq_hat_B_lb}-\eqref{eq_hat_B_ub}, for commodity $m$ satisfying $\hat B_{ij,m}^{\text{lb}}(t) \le \hat B_{ij,m}^{\text{ub}}(t)$, we have 
\begin{equation}
	\Scale[1]{W_{ij}(x) = \Delta U_{ij}^{m}(t) + x\Delta L_{ij}^{m},\ \ \ \forall x \in [\hat B_{ij,m}^{\text{lb}}(t), \hat B_{ij,m}^{\text{ub}}(t)]}.\label{eq_W_x}
\end{equation}
Based on \eqref{eq_W_x}, we upper bound function $Z_x({\bf b}_{ij}(t))$ in problem \eqref{eq_sdado-routing_opt} for $x \in [\hat B_{ij,m}^{\text{lb}}(t), \hat B_{ij,m}^{\text{ub}}(t)]$,
\begin{align}
	\Scale[1]{Z_x({\bf b}_{ij}(t))}&\Scale[1]{=\sum\nolimits_{m'}{\left[\Delta U_{ij}^{m'}(t)+x\Delta L_{ij}^{m'}\right]z^{m'} b_{ij}^{m'}(t)},}\notag\\
	&\Scale[1]{\le \left[\Delta U_{ij}^{m}(t) + x\Delta L_{ij}^{m}\right] \sum\nolimits_{m'}{z^{m'} b_{ij}^{m'}(t)}}\notag\\ 
	&\Scale[1]{= W_{ij}(x)\sum\nolimits_{m'}{z^{m'} b_{ij}^{m'}(t)}},\notag\\
	&\Scale[1]{\le \left[W_{ij}(x)\right]^+ \bar C_{ji},}\label{eq_bound_Z}
\end{align}
where the upper bound is achieved if $W_{ij}(x)\ge 0$, $b_{ij}^m(t)=\nicefrac{\bar C_{ji}}{z^m}$, and $b_{ij}^{m'}(t)=0, \forall m'\neq m$. Denote $Z_x^*\triangleq\max_{{\bf b}_{ij}(t)}{Z_x({\bf b}_{ij}(t))}$ subject to \eqref{eq_sdado_routing_cap_constraint}, and it follows from \eqref{eq_bound_Z} that $Z_x^*=\left[W_{ij}(x)\right]^+ \bar C_{ji}$. Moreover, if $x=\eta_{ij}(t)$, we have $Z_{\eta_{ij}(t)}^*=Z_{\eta_{ij}(t)}({\bf b}_{ij}^*(t))$.

With Assumptions \ref{assume_L_quantitization} and \ref{assume_integer_multiple}, we run SDADO-interest-forwarding according to Algorithm \ref{alg_fowarding}-\ref{alg_FwdRateAsgmt} and obtain the following conclusions.
\begin{itemize}
	\item When Condition \ref{cond_i}) is satisfied, we have $W_{ij}(x) > 0$; $b_{ij}^m(t)=\nicefrac{\bar C_{ji}}{z^m}$; $b_{ij}^{m'}(t)=0, \forall m'\neq m$; $\eta_{ij}(t)\in [ B_{ij,m}^{\text{lb}}(t);  B_{ij,m}^{\text{ub}}(t)] \subseteq [\hat B_{ij,m}^{\text{lb}}(t), \hat B_{ij,m}^{\text{ub}}(t)]$. Then it follows from \eqref{eq_bound_Z} that
	\begin{equation}
		\Scale[1]{Z_{\eta_{ij}(t)}({\bf b}_{ij}(t))=Z_{\eta_{ij}(t)}^*.}\label{eq_Z_vs_Z*1}
	\end{equation}
	
	\item When Condition \ref{cond_ii}) is satisfied, we have $b_{ij}^{m'}(t)=0, \forall m'$; $\eta_{ij}(t)\in [\hat B_{ij,m}^{\text{lb}}(t), \hat B_{ij,m}^{\text{ub}}(t)]$; $\mathbbm{1}(\Delta L_{ij}^{m}> 0)\mathbbm{1}(\Delta U_{ij}^{m}(t) \ge -h_i) + \mathbbm{1}(\Delta L_{ij}^{m}\le 0)\mathbbm{1}(\Delta U_{ij}^{m}(t)\ge h_i)=0$. For commodity $m$, based on \eqref{eq_h_ij_bound} and the fact $ \Delta L_{ij}^m \le l_{ij}$, we have, $\forall \kappa \in \mathbb{R}^+$,
	\begin{align}
		&\Scale[1]{\Delta U_{ij}^{m}(t) + \eta_{ij}(t)\Delta L_{ij}^{m}} \notag\\
		&\Scale[1]{\!\!\le\!\begin{cases}
			-h_i + \left[\frac{h_i }{\delta_L}+\kappa\right]l_{ij}, & \text{if $\Delta L_{ij}^m>0$},\\
			h_i, & \text{if $\Delta L_{ij}^m\le 0$.}
		\end{cases}}\label{eq_bound_condition4}
	\end{align}
	With the fact $l_{ij}\ge \delta_L$, we further plug \eqref{eq_bound_condition4} into \eqref{eq_sdado_routing_obj} and upper bound $Z_{\eta_{ij}(t)}^*$:
	\begin{equation}
		\Scale[1]{Z_{\eta_{ij}(t)}^* \le \left[\frac{l_{ij}}{\delta_L} - 1\right]h_i\bar C_{ji} + \kappa l_{ij}\bar C_{ji},\ \ \forall \kappa \in \mathbb{R}^+}.
	\end{equation}
	Since $Z_{\eta_{ij}(t)}({\bf b}_{ij}(t))=0$ due to $b_{ij}^{m'}(t)=0, \forall m'$, it follows that
	\begin{align}
		\Scale[1]{Z_{\eta_{ij}(t)}({\bf b}_{ij}(t)) \ge Z_{\eta_{ij}(t)}^* - \left[\frac{l_{ij}}{\delta_L} - 1\right]}&\Scale[1]{h_i\bar C_{ji}-\kappa l_{ij}\bar C_{ji},}\notag\\
		&\ \ \ \Scale[1]{\forall \kappa \in \mathbb{R}^+.}\label{eq_Z_vs_Z*2}
	\end{align}
\end{itemize}

By setting $\kappa = 0$, it follows from \eqref{eq_h_ij_bound} that $h_{ij}(t) \le \frac{h_i}{\delta_L}$, and we have $Z_{\eta_{ij}(t)}({\bf b}_{ij}(t)) \ge Z_{\eta_{ij}(t)}^* - \left[\nicefrac{l_{ij}}{\delta_L} - 1\right]h_i\bar C_{ji}$ in summary of \eqref{eq_Z_vs_Z*1} and \eqref{eq_Z_vs_Z*2}.

%

\section{Proof of Theorem \ref{thm_stability}}
\label{appendix_thm_proof}
Let ${\bf U}(t)=\{U_i^{\phi,k,c}(t)\}$ represent the vector of virtual queue backlog values of all the commodities at all the network nodes. The network \emph{Lyapunov Drift} (LD) is defined as
\begin{equation}
	\Scale[1]{\Delta \left({\bf U}(t)\right) \triangleq \frac{1}{2} \mathbb{E} \left\{\left\|{\bf U}(t+1)\right\|^2 - \left\|{\bf U}(t)\right\|^2\right\}},
\end{equation}
where $\|*\|$ indicates Euclidean norm, and the expectation is taken over the ensemble of all the realizations of service request generations at consumers.

After multiplying $z^{(\phi,k)}$ and then squaring both sides of \eqref{eq_queue_dynamic}, by following standard LD manipulations (see reference \cite{bib_Neely_book}), we upper bound LD as
\begin{align}
	\Scale[1]{\Delta \left({\bf U}(t)\right) \le \mathbb{E}\left\{\left. \Gamma(t) + Y(t)\right|{\bf U}(t)\right\} + {\bm \chi}^{\dagger} {\bf U}(t)},\label{eq_LD_1}
\end{align}
where
\begin{align}
	\Scale[0.95]{\Gamma(t) \!\triangleq} & \Scale[0.95]{\frac{1}{2}\sum_{i}\sum_{(\phi,k,c)}\!\left[z^{(\phi,k)}\right]^2\!\!\left\{\!\left[\sum_{j\in {\mathcal O}(i)}\!b_{ij}^{(\phi,k,c)}\!(t) \!+\! p_i^{(\phi,k,c)}\!(t)\right]^{2}\right.}\notag\\
	& \Scale[0.95]{+\!\! \left.\left[\sum_{j\in {\mathcal O}(i)}b_{ji}^{(\phi,k,c)}\!(t) \!+\! p_i^{(\phi,k+1,c)}\!(t) \!+\! a_i^\phi\! (t){\mathbbm 1}\!(k\!=\!K_\phi)\right]^2\!\right\}};\notag\\
	\Scale[0.95]{Y(t)\triangleq} & \Scale[0.95]{\ \sum_{i}\sum_{(\phi,k,c)} U_i^{(\phi,k,c)}(t)z^{(\phi,k)} \left[\sum_{j\in {\mathcal O}(i)}b_{ji}^{(\phi,k,c)}(t)\right.}\notag\\
	& \Scale[0.95]{\left.+ p_i^{(\phi,k+1,c)}(t) - \sum_{j\in {\mathcal O}(i)}b_{ij}^{(\phi,k,c)}(t) - p_i^{(\phi,k,c)}(t)\right].}\notag
\end{align}
By denoting $C_{\text{max}}^{\text{tr}-}\triangleq \max_i\{\sum_{j}\bar C_{ij}\}$, $C_{\text{max}}^{\text{tr}+}\triangleq \max_i\{\sum_{j}\bar C_{ji}\}$, $l_{\text{max}}\triangleq\max_{(i,j)}\{l_{ij}\}$, $h_{\text{max}}=\max_i\{h_i\}$, $z_{\text{max}}\triangleq \max_{m}\{z^m\}$, $r_{\text{min}}\triangleq \min_m\{r^m\}$, $C_{\text{max}}^{\text{pr}}\triangleq \max_i\{C_i^{\text{pr}}\}$, $C_{\text{max}}^{\text{dp}}\triangleq \max_i\{C_i^{\text{dp}}\}$, we upper bound $\Gamma(t)$ as follows:
\begin{align}
	\Scale[0.95]{\Gamma(t)} &\Scale[0.95]{\le\!\frac{\left|\mathcal V\right|}{2} \left\{\!\left[C_{\text{max}}^{\text{tr}+} \!+\! \frac{z_\text{max}C_{\text{max}}^{\text{pr}}}{r_{\text{min}}}\!+\!C_{\text{max}}^{\text{dp}}\right]^2 \!+\! 
		\left[C_{\text{max}}^{\text{tr}-} \!+\! \frac{z_\text{max}C_{\text{max}}^{\text{pr}}}{r_{\text{min}}} \!+\! A_{\text{max}}\right]^2 \!\right\}} \notag\\
	& \Scale[0.95]{\triangleq B_0,}\notag
\end{align}
and decompose $Y(t)$ as follows:
\begin{align}
	Y(t) = Y_{\text{dp}}(t) + Y_{\text{pr}}(t) + Y_{\text{tr}}(t),
\end{align}
where 
\begin{align}
	\Scale[0.95]{Y_{\text{dp}}(t)} &\Scale[0.95]{\triangleq -\sum\nolimits_i\sum\nolimits_{(\phi,c): i\in {\mathcal P}^{(\phi,0)}} U_{i}^{(\phi,0,c)}(t)z^{(\phi,0)}p_i^{(\phi,0,c)}(t);}\notag\\
	\Scale[0.95]{Y_{\text{pr}}(t)} &\Scale[0.95]{\triangleq -\sum\nolimits_i\sum\nolimits_{(\phi,k>0,c): i\in {\mathcal P}^{(\phi,k,c)}}p_i^{(\phi,k,c)}(t)}\notag\\
	&\ \ \  \ \  \ \  \ \ \ \ \ \  \Scale[0.95]{\times\left[U_i^{(\phi,k,c)}(t)z^{(\phi,k)}-U_i^{(\phi,k-1,c)}(t)z^{(\phi,k-1)}\right];}\notag\\
	\Scale[0.95]{Y_{\text{tr}}(t)} &\Scale[0.95]{\triangleq -\sum_i \sum_{j\in \mathcal O(i)} \sum_{(\phi,k,c)} b_{ij}^{(\phi,k,c)}(t) z^{(\phi,k)}}\notag\\
	&\ \ \ \ \ \ \ \ \ \ \ \ \ \ \ \ \ \ \ \ \ \ \ \ \ \ \ \ \ \ \ \Scale[0.95]{\times\left[U_i^{(\phi,k,c)}(t)-U_j^{(\phi,k,c)}(t)\right]}\notag\\
	&\Scale[0.95]{=-\sum_i \sum_{j\in \mathcal O(i)} \sum_{(\phi,k,c)} b_{ij}^{(\phi,k,c)}(t) z^{(\phi,k)}\Delta U_{ij}^{(\phi,k,c)}(t).}\notag
\end{align}

By adding
\begin{equation}
	\Scale[0.95]{Y_L(t) = -\sum_i \sum_{j\in \mathcal O(i)} \eta_{ij}(t)\sum_{(\phi,k,c)} b_{ij}^{(\phi,k,c)}(t) z^{(\phi,k)}\Delta L_{ij}^{(\phi,k,c)}}\notag
\end{equation}
onto both sides of \eqref{eq_LD_1}, we obtain
\begin{align}
	&\Scale[1]{\Delta \left({\bf U}(t)\right) + \mathbb{E}\left\{\left.Y_L(t)\right|{\bf U}(t)\right\}}\notag\\
	&\Scale[1]{\le \mathbb{E}\left\{\left. \Gamma(t) + Y(t)+Y_L(t)\right|{\bf U}(t)\right\} + {\bm \chi}^{\dagger} {\bf U}(t)}\notag\\
	&\Scale[1]{\le B_0 + \mathbb{E}\left\{\left. Y_{\text{dp}}(t)+Y_{\text{pr}}(t)+ \tilde Y_{\text{tr}}(t)\right|{\bf U}(t)\right\} + {\bm \chi}^{\dagger} {\bf U}(t),}\label{eq_LD_2}
\end{align}
where $\eta_{ij}(t)$ is set according to \eqref{eq_h_1} and \eqref{eq_h_2} with $\kappa=0$, and
\begin{align}
	\Scale[1]{\tilde Y_{\text{tr}}(t)} &\Scale[1]{\triangleq Y_{\text{tr}}(t)+Y_L(t)}\notag\\
	&\Scale[1]{=-\sum_i \sum_{j\in \mathcal O(i)} \sum_{(\phi,k,c)} b_{ij}^{(\phi,k,c)}(t) z^{(\phi,k)}}\notag\\
	&\ \ \ \ \ \ \ \ \ \ \ \ \ \ \Scale[1]{\times \left[\Delta U_{ij}^{(\phi,k,c)}(t) + \eta_{ij}(t)\Delta L_{ij}^{(\phi,k,c)}\right].}\label{eq_Y_Z}
\end{align}

On the other hand, given ${\bm \chi}$ being interior to the computing network capacity region $\Lambda$, there exists a positive number $\epsilon$ such that ${\bm \chi} + \epsilon {\bf 1}\in \Lambda$. According to the theorem of computing network capacity region in reference \cite{bib_DCNC}, there exists a stationary randomized policy $\pi^\star$ that determines the $b_{ij}^{(\phi,k,c)\star}(t)$ and $p_i^{(\phi,k,c)\star}(t)$ such that, $\forall i,(\phi,k,c),t$,
\begin{align}
	&\Scale[0.95]{\mathbb{E}\left\{\sum_{j\in {\mathcal O}(i)}b_{ji}^{(\phi,k,c)\star}(t)+ p_i^{(\phi,k+1,c)\star}(t)- \sum_{j\in {\mathcal O}(i)}b_{ij}^{(\phi,k,c)\star}(t)\right.}\notag\\
	&\ \ \ \Scale[0.95]{\left. - p_i^{(\phi,k,c)\star}(t)\right\}z^{(\phi,k)}\le -\lambda_i^\phi z^{(\phi,k)} \mathbbm{1}(k=K_\phi)-\epsilon.}\label{eq_stationary_policy}
\end{align}

Going back to \eqref{eq_LD_2}, with Assumption \ref{assume_integer_multiple}, minimizing $Y_{\text{dp}}(t)$ subject to \eqref{eq_produce_capacity} and minimizing $Y_{\text{pr}}(t)$ subject to \eqref{eq_compute_capacity} both result in max-weight-matching solutions whose implementations are SDADO-data-producing and -processing-commitment shown in Algorithm \ref{alg_data_producing} and \ref{alg_proc_commitment}, respectively. In addition, according to Lemma \ref{lemma_sdado_routing} with Assumption \ref{assume_L_quantitization} and \ref{assume_integer_multiple}, equation \eqref{eq_Y_Z} results in that $\tilde Y_{\text{tr}}(t)$ under SDADO-interest-forwarding in Algorithm \ref{alg_fowarding}-\ref{alg_FwdRateAsgmt} satisfies
\begin{align}
	\!\!\!\Scale[1]{\tilde Y_{\text{tr}}(t)} &\Scale[1]{=-\sum_i \sum_{j\in \mathcal O(i)} Z_{\eta_{ij}(t)}({\bf b}_{ij}(t)).}\notag\\
	&\Scale[1]{\le \sum_i \sum_{j\in \mathcal O(i)} \left[-Z_{\eta_{ij}(t)}^* + \left[\frac{l_{ij}}{\delta_L}-1\right]h_i \bar C_{ji}\right]}\notag\\
	&\Scale[1]{\le -\sum_i \sum_{j\in \mathcal O(i)} Z_{\eta_{ij}(t)}^* \!+\! \left|\mathcal V\right|\left[\frac{l_{\text{max}}}{\delta_L}-1\right]h_{\text{max}}C_{\text{max}}^{\text{tr}+},}\notag\\
	&\Scale[1]{\triangleq -\sum_i \sum_{j\in \mathcal O(i)} Z_{\eta_{ij}(t)}^* + B_{\text{tr}}}
\end{align}
where $-\sum_i \sum_{j\in \mathcal O(i)} Z_{\eta_{ij}(t)}^*$ is equal to the minimized $\tilde Y_{\text{tr}}(t)$ subject to \eqref{eq_comm_capacity}, and $B_{\text{tr}} \triangleq \left|\mathcal V\right|\left[\nicefrac{l_{\text{max}}}{\delta_L}-1\right]h_{\text{max}}C_{\text{max}}^{\text{tr}+}$. 

Then it follows from \eqref{eq_LD_2} and \eqref{eq_stationary_policy} that the SDADO solution in Algorithm \ref{alg_data_producing}-\ref{alg_FwdRateAsgmt} satisfy 
\begin{align}
	&\Scale[0.95]{\Delta \left({\bf U}(t)\right) + \mathbb{E}\left\{\left.Y_L(t)\right|{\bf U}(t)\right\}}\notag\\
	&\Scale[0.95]{\le B_0 + B_{\text{tr}}+\mathbb{E}\left\{\left. Y_{\text{dp}}^\star(t)+Y_{\text{pr}}^\star(t)+ \tilde Y_{\text{tr}}^\star(t)\right|{\bf U}(t)\right\} + {\bm \chi}^{\dagger} {\bf U}(t)}\notag\\
	&\Scale[0.95]{= B_0 + B_{\text{tr}} +\mathbb{E}\left\{\left. Y^\star(t) + Y_L^\star(t)\right|{\bf U}(t)\right\} \!+\! {\bm \chi}^{\dagger} {\bf U}(t)}\notag
\end{align}
\begin{align}
	&\Scale[0.95]{=B_0 + B_{\text{tr}}+\sum_{i}\sum_{(\phi,k,c)} U_i^{(\phi,k,c)}(t)z^{(\phi,k)}\mathbb{E}\left\{p_i^{(\phi,k+1,c)\star}(t)\right.}\notag\\
	&\ \ \ \Scale[0.95]{\left.+ \sum_{j\in {\mathcal O}(i)}b_{ji}^{(\phi,k,c)\star}(t) - \sum_{j\in {\mathcal O}(i)}b_{ij}^{(\phi,k,c)\star}(t) - p_i^{(\phi,k,c)\star}(t)\right\}}\notag\\
	& \ \ \ \Scale[0.95]{+\sum_{i}\sum_{(\phi,c)} U_i^{(\phi,0,c)}(t)z^{(\phi,0)}\lambda_i^\phi + \mathbb{E}\left\{\left.Y_L^\star(t)\right|{\bf U}(t)\right\}}\notag\\
	&\Scale[0.95]{\le B_0 + B_{\text{tr}}- \sum_{i}\sum_{(\phi,k,c)}U_i^{(\phi,k,c)}(t)\left[\lambda_i^\phi z^{(\phi,k)} \mathbbm{1}(k=K_\phi)+\epsilon\right]}\notag\\
	& \ \ \ \Scale[0.95]{+\sum_{i}\sum_{(\phi,c)} U_i^{(\phi,K_\phi,c)}(t)z^{(\phi,K_\phi)}\lambda_i^\phi + \mathbb{E}\left\{\left.Y_L^\star(t)\right|{\bf U}(t)\right\}}\notag\\
	&\Scale[0.95]{\le B_0 + B_{\text{tr}} - \epsilon \sum_{i}\sum_{(\phi,k,c)}U_i^{(\phi,k,c)}(t)  + \mathbb{E}\left\{\left.Y_L^\star(t)\right|{\bf U}(t)\right\}}\label{eq_LD_3}
\end{align}

Moreover, using \eqref{eq_h_ij_bound} with $\kappa = 0$ and the fact $-\Delta L_{ij}^{(\phi,k,c)}=\Delta L_{ji}^{(\phi,k,c)}$, we further upper bound $-Y_L(t)$ and $Y_L^\star(t)$ as follows:
\begin{align}
	\Scale[1]{-Y_L(t)\le \sum_{i}\frac{h_i}{\delta_L}l_{ij}\bar C_{ji}\le  \frac{1}{\delta_L}\left|\mathcal V\right|h_{\text{max}}l_{\text{max}}C_{\text{max}}^{\text{tr}+}\triangleq B_L};\label{eq_Y_L_bound1}\\
	\Scale[1]{Y_L^\star(t)\le \sum_{i}\frac{h_i}{\delta_L}l_{ji}\bar C_{ji}\le  \frac{1}{\delta_L}\left|\mathcal V\right|h_{\text{max}}l_{\text{max}}C_{\text{max}}^{\text{tr}+} = B_L.}\label{eq_Y_L_bound2}
\end{align}
Then, by plugging \eqref{eq_Y_L_bound1} and \eqref{eq_Y_L_bound2} into \eqref{eq_LD_3}, we have
\begin{align}
	\Scale[1]{\Delta \left({\bf U}(t)\right)}  &\Scale[1]{\le B_0 + B_{\text{tr}}+\mathbb{E}\left\{\left.Y_L^\star(t)-Y_L(t)\right|{\bf U}(t)\right\}}\notag\\
	&\ \ \ \Scale[1]{- \epsilon \sum_{i}\sum_{(\phi,k,c)}U_i^{(\phi,k,c)}(t)}\notag\\
	&\Scale[1]{\le B_0 + B_{\text{tr}} + 2 B_L - \epsilon \sum_{i}\sum_{(\phi,k,c)}U_i^{(\phi,k,c)}(t)}\notag\\
	&\Scale[1]{\triangleq B_0 + B_1 - \epsilon \sum_{i}\sum_{(\phi,k,c)}U_i^{(\phi,k,c)}(t),}\label{eq_LD_4}
\end{align}
where $B_1\triangleq B_{\text{tr}} + 2 B_L$.

We then use the theoretical results in \cite{bib_Neely_prob_1} for the proof of network stability with probability $1$. Note that the following bounding condition is satisfied: for $\forall i, (\phi,k,c),t$,
\begin{align}
	&\Scale[1]{\mathbb{E}\left\{\left.\left[U_i^{(\phi,k,c)}(t+1)-U_i^{(\phi,k,c)}(t)\right]^4\right|{\bf U}(t)\right\}}\notag\\
	&\Scale[1]{\le \left[C_{\text{max}}^{\text{tr}-} + \frac{z_\text{max}C_{\text{max}}^{\text{pr}}}{r_{\text{min}}} \!+\! A_{\text{max}}\right]^4 <+\infty}\label{eq_4th_moment_bound}
\end{align}  
Then with \eqref{eq_4th_moment_bound} satisfied, equation \eqref{eq_LD_4} leads to the following stability conclusion based on the derivations in \cite{bib_Neely_prob_1}:
\begin{equation}
	{\limsup\limits_{t\rightarrow \infty} \frac{1}{t}\sum\nolimits_{\tau, m, i} U_i^m(\tau) \le \frac{B_0 + B_1}{\epsilon},}\ \ \text{with prob. 1.}\notag
\end{equation}

\bibliographystyle{IEEEtran}
\bibliography{Reference_long}



\end{document}